\documentclass[journal]{IEEEtran}
\usepackage{amsmath,amsfonts}
\usepackage{algorithmic}
\usepackage{algorithm}
\usepackage{array}
\usepackage[caption=false,font=normalsize,labelfont=sf,textfont=sf]{subfig}
\usepackage{textcomp}
\usepackage{stfloats}
\usepackage{url}
\usepackage{verbatim}
\usepackage{graphicx}
\usepackage{cite}
\usepackage{balance}

\usepackage{amssymb, amsthm}
\newtheorem{theorem}{Theorem}
\newtheorem{proposition}{Proposition}
\newtheorem{lemma}{Lemma}

\begin{document}

\title{Federated Calibration of Motion Uncertainty for UAV Tracking in Multi-Operator ISAC Systems}

\author{Qiming~Li,~\IEEEmembership{Graduate~Student~Member,~IEEE,}
        Lei~Zhang,~\IEEEmembership{Senior~Member,~IEEE,}
        Muhammad~Ali~Imran,~\IEEEmembership{Fellow,~IEEE}
        and~Lina~Mohjazi,~\IEEEmembership{Senior~Member,~IEEE}

\thanks{Qiming~Li, Lei~Zhang, Muhammad~Ali~Imran and Lina~Mohjazi are with the James Watt School of Engineering, University of Glasgow, G12 8QQ Glasgow U.K. (e-mail: q.li.5@research.gla.ac.uk; Lei.Zhang@glasgow.ac.uk; Muhammad.Imran@glasgow.ac.uk; Lina.Mohjazi@glasgow.ac.uk).}
}

\maketitle

\begingroup
\renewcommand\thefootnote{}
\footnotetext{This work has been submitted to the IEEE for possible publication.
Copyright may be transferred without notice, after which this version may no longer be accessible.}
\addtocounter{footnote}{-1}
\endgroup

\begin{abstract}
In integrated sensing and communication (ISAC) systems involving multiple operators, unmanned aerial vehicle (UAV) tracking can benefit from sensing at multiple base stations, while raw sensing records, channel information, and transmit decisions remain local to each operator. The resulting problem is to obtain a reliable UAV prediction from different track estimates and use it to guide transmission without centralizing their data or control variables. To resolve this, we propose a two-timescale framework comprising slow-timescale motion-uncertainty calibration and online track fusion with predictive beamforming. Local radar innovations calibrate the uncertainty assigned by each tracker to unmodeled UAV motion, with only the resulting parameters shared across operators. During online tracking, covariance intersection (CI) combines local estimates without requiring knowledge of correlations between their estimation errors. The fused prediction then guides each operator’s transmission design subject to communication QoS, power, and tracking constraints. We further establish exact rank one recovery for the communication covariances and finite interval bounds on local tracking uncertainty. Simulations show improved consistency of tracking uncertainty under substantial motion model mismatch. CI avoids overconfidence under correlated errors, while using the current fused prediction reduces RMSE by about 31\% relative to transmission designed for communication alone, with less than 2\% additional power. Ablation further shows that the power difference between current and outdated predictions is driven mainly by predictive covariance rather than predictive mean.

\end{abstract}


\section{Introduction}
\label{sec:introduction}

\IEEEPARstart{A}{ccurate} knowledge of the unmanned aerial vehicle (UAV) state is important for sensing-aware transmission in future wireless networks, while communication services must be maintained simultaneously. Integrated sensing and communication (ISAC) provides a natural framework for supporting both functions \cite{mozaffari2019uavTutorial,liu2022isacSurvey}. In such systems, UAV tracking and transmission design are inherently coupled: the estimated UAV state affects the transmission design, while the selected transmission determines the quality of subsequent sensing observations. Tracking and transmission should therefore be considered in a closed loop \cite{liu2020radar}.

Spatially separated base stations (BSs) can provide complementary sensing views and improve tracking availability~\cite{cheng2024networkedISAC}, while sensing and communication (S\&C) links may become unavailable under dynamic blockage \cite{al2024enhancing}. In areas covered by multiple operators, BSs belonging to different networks could provide a potential source of such spatial diversity without requiring additional sensing sites. Exploiting this diversity, however, should preserve the operational independence of the participating operators. Existing network-sharing arrangements show that infrastructure can be shared while operators retain separate network control \cite{khan2011networkSharing}. Accordingly, we consider limited inter-operator exchange of tracking information rather than joint radio-resource management: raw radar measurements, CSI, user information, and transmission decisions remain within each operator.

Within these information boundaries, exchanging only a predictive mean of the UAV position and velocity is insufficient for transmission design; operators also need to know how reliable that prediction is. The predicted UAV position determines the nominal sensing geometry, including range, angle, path loss, and array response, while the associated uncertainty determines how widely directional sensing power should be distributed around the predicted UAV position \cite{meng2023vehicular}. An underestimated uncertainty may concentrate sensing power too narrowly around the predicted position, whereas overestimated uncertainty may unnecessarily increase transmit power. However, obtaining reliable and current prediction uncertainty involves three distinct difficulties. First, the motion model used by a local tracker may not match the UAV’s actual maneuvers, causing its uncertainty estimate to become inaccurate \cite{li2003surveyManeuveringTargets,huang2018novel}. Second, local trackers may have correlated estimation errors because they share prior information and common target dynamics, while the exact correlations are unknown \cite{julier1997nonDivergent,chen2002ciRevisited}. Finally, even a well calibrated prediction becomes outdated if it does not incorporate the latest sensing observations \cite{bar2001estimation,liu2020radar}. The central question is therefore how to provide operators with a reliable and up to date UAV prediction, together with its uncertainty, while keeping radar histories and radio control variables local.

To address these challenges, we develop a two-timescale closed-loop framework. At the slower time scale, each operator uses local radar innovations to calibrate the uncertainty associated with unmodeled UAV motion, while only the resulting calibration parameters are shared. During online tracking, covariance intersection (CI) combines the local estimates without requiring knowledge of the correlations between their estimation errors \cite{chen2002ciRevisited}. The fused UAV state estimate and covariance are then propagated one slot ahead, and the resulting predictive mean and covariance are used with each operator's local CSI, communication requirements, and power budget to design its next transmission.

\begin{table*}[t]
\centering
\caption{Comparison with representative related works.}
\label{tab:related_work}
\footnotesize
\renewcommand{\arraystretch}{1.05}

\begin{tabular*}{\textwidth}{@{\extracolsep{\fill}}lcccc@{}}
\hline
Work &
Uncertainty treatment &
Track fusion &
Transmission coupling &
Coordination scope\\
\hline

Liu \emph{et al.} \cite{liu2020radar} &
Filter prediction &
-- &
Predictive beamforming &
Single system \\

Cheng \emph{et al.} \cite{cheng2024networkedISAC} &
-- &
-- &
Coordinated beamforming &
Multi-BS \\

Meng \emph{et al.} \cite{meng2023vehicular} &
Model-aware tracking &
-- &
Beam adaptation &
Single system \\

Huang \emph{et al.} \cite{huang2018novel} &
Adaptive covariance &
-- &
-- &
-- \\

Chen \emph{et al.} \cite{chen2002ciRevisited} &
-- &
CI &
-- &
Distributed estimation \\

\textbf{This work} &
\textbf{Federated calibration} &
\textbf{CI} &
\textbf{Predictive covariance-aware design} &
\textbf{Multi-operator} \\
\hline
\end{tabular*}
\end{table*}

\subsection{Related Work}
\label{subsec:related_work}

The literature most relevant to this work spans networked ISAC transmission design, maneuvering-target tracking, covariance calibration and track fusion, and federated optimization. Signal-processing treatments establish the waveform, channel, and receiver foundations for ISAC \cite{zhang2021jcrOverview}, while system-level studies characterize the tradeoff between S\&C objectives \cite{liu2020jointRadarComm,li2026distributed}. Optimization-based designs make sensing accuracy explicit through Cram\'{e}r-Rao bound criteria\cite{liu2022crbJrc,li2026distributed}, and networked ISAC extends this principle to coordinated transmission across multiple BSs under S\&C constraints \cite{cheng2024networkedISAC}. Related multicell MIMO studies develop coordinated beamforming under downlink QoS requirements \cite{gesbert2010multicellMIMO,dahrouj2010coordinatedBeamforming}. Convex precoding yields tractable SINR-constrained formulations \cite{wiesel2006linearPrecoding}, whereas worst-case designs protect QoS against bounded channel uncertainty \cite{vucic2009robustQoS}. These works provide the transmission-design foundations used in this paper. Their application to target-aware control requires sufficiently reliable knowledge of the target state and its future evolution, which motivates the following discussion of maneuvering-target tracking.

Maneuvering-target tracking has been studied using standard kinematic models, multiple-model methods such as interacting multiple model filtering, adaptive approaches, and more recently learning-based predictors \cite{li2003surveyManeuveringTargets,rong2005survey,huang2018novel,liu2023digital}. These methods provide increasingly flexible representations of target motion. In contrast, this work retains a fixed nearly constant velocity (NCV) model as a common low-dimensional tracking surrogate and calibrates the residual process uncertainty caused by unmodeled maneuvers. The objective is not to replace richer motion predictors, but to obtain a common Gaussian state and covariance representation that can be calibrated locally, fused across operators, and used in downstream transmission design. This setting also differs from UAV trajectory optimization, where the UAV path itself is a network control variable \cite{jing2024isac}.

Reliable use of such predictions further requires meaningful uncertainty estimates and appropriate fusion of distributed tracks. Innovation-based covariance identification and adaptive covariance estimation have been extensively studied in \cite{huang2018novel,mehra1970identification,odelson2006autocovariance}, while distributed Kalman filtering combines information across sensor networks \cite{olfatiSaber2007distributedKalman}. When inter-track error cross-correlations are unavailable, CI provides a conservative Gaussian fusion rule without explicitly reconstructing them \cite{julier1997nonDivergent,chen2002ciRevisited}. These methods address uncertainty estimation or track fusion, but do not directly connect the resulting predictive covariance to multi-operator ISAC transmission control.

Federated learning (FL), exemplified by FedAvg \cite{mcmahan2017fedAvg}, aggregates locally computed updates without centralizing the underlying training data \cite{mcmahan2017fedAvg}, and wireless FL exposes the communication and deployment constraints of such coordination \cite{niknam2020wirelessFL}. Heterogeneous local objectives motivate proximal federated optimization \cite{li2020fedprox}, while personalized FL distinguishes globally shared structure from client-specific models \cite{tan2023personalizedFL}. Here, \textit{federated} refers specifically to the local estimation and aggregation of a low-dimensional process-covariance parameter rather than to the training of a trajectory-prediction model. Radar histories and local innovations remain at the individual operators. Table~\ref{tab:related_work} summarizes representative related works, highlighting how this work integrates federated uncertainty calibration, correlation-aware track fusion, and predictive transmission design while keeping radio control local to each operator.

\subsection{Contributions}
\label{subsec:contributions}

The main contribution is their integration into a two-timescale closed-loop ISAC framework under operator-local data and radio control. Specifically:

\begin{itemize}
    \item \textbf{Federated motion-uncertainty calibration:}
    We develop a procedure for calibrating the process uncertainty of a fixed NCV tracker from locally retained radar data, without sharing measurement histories across operators. Each operator estimates the residual process covariance from its local radar innovations, while only a low-dimensional calibration parameter is exchanged for aggregation. The resulting covariance is then fixed for online prediction, providing a common uncertainty model that accounts for motion not captured by the nominal tracker.

    \item \textbf{Track fusion under unknown correlations:}
    We combine synchronized Gaussian track summaries from different operators without requiring their error cross-correlations to be known or exchanged. Weighted covariance intersection is used to obtain the fused state and covariance, which are subsequently propagated for transmission design. We further derive a CI-induced covariance bound relation that connects the fused and local covariances and supports the finite-interval analysis.

    \item \textbf{Predictive uncertainty-aware beamforming:}
    We incorporate the fused one-step predictive mean and covariance into each operator's local transmit-covariance design under communication signal-to-interference-plus-noise ratio (SINR), power, and sensing requirements. For the free residual covariance formulation, we establish exact rank-one recovery of the communication covariances and derive conditional finite-interval bounds that relate the tracking constraint in the transmission design to local tracker covariances.
\end{itemize}

\noindent\textit{Notations:}
Scalars, vectors, and matrices are denoted by lowercase, bold lowercase, and bold uppercase letters, respectively. Sets are denoted by calligraphic letters. The operators $(\cdot)^{\mathsf{T}}$ and $(\cdot)^{\mathsf{H}}$ denote transpose and Hermitian transpose. The Euclidean norm is denoted by $\|\cdot\|_2$. For a matrix $\mathbf A$, $|\mathbf A|$, $\operatorname{tr}(\mathbf A)$, and $\mathbf A\succeq\mathbf0$ denote determinant, trace, and positive semidefiniteness, respectively. The real and complex fields are denoted by $\mathbb R$ and $\mathbb C$. Unless otherwise specified, indices $c$, $b$, $k$, and $i$ refer to an operator, a BS, a communication user, and a sensing design point, respectively.

\section{System Model}
\label{sec:system_model}

\begin{figure*}[!t]
\centering
\includegraphics[width=0.9\textwidth]{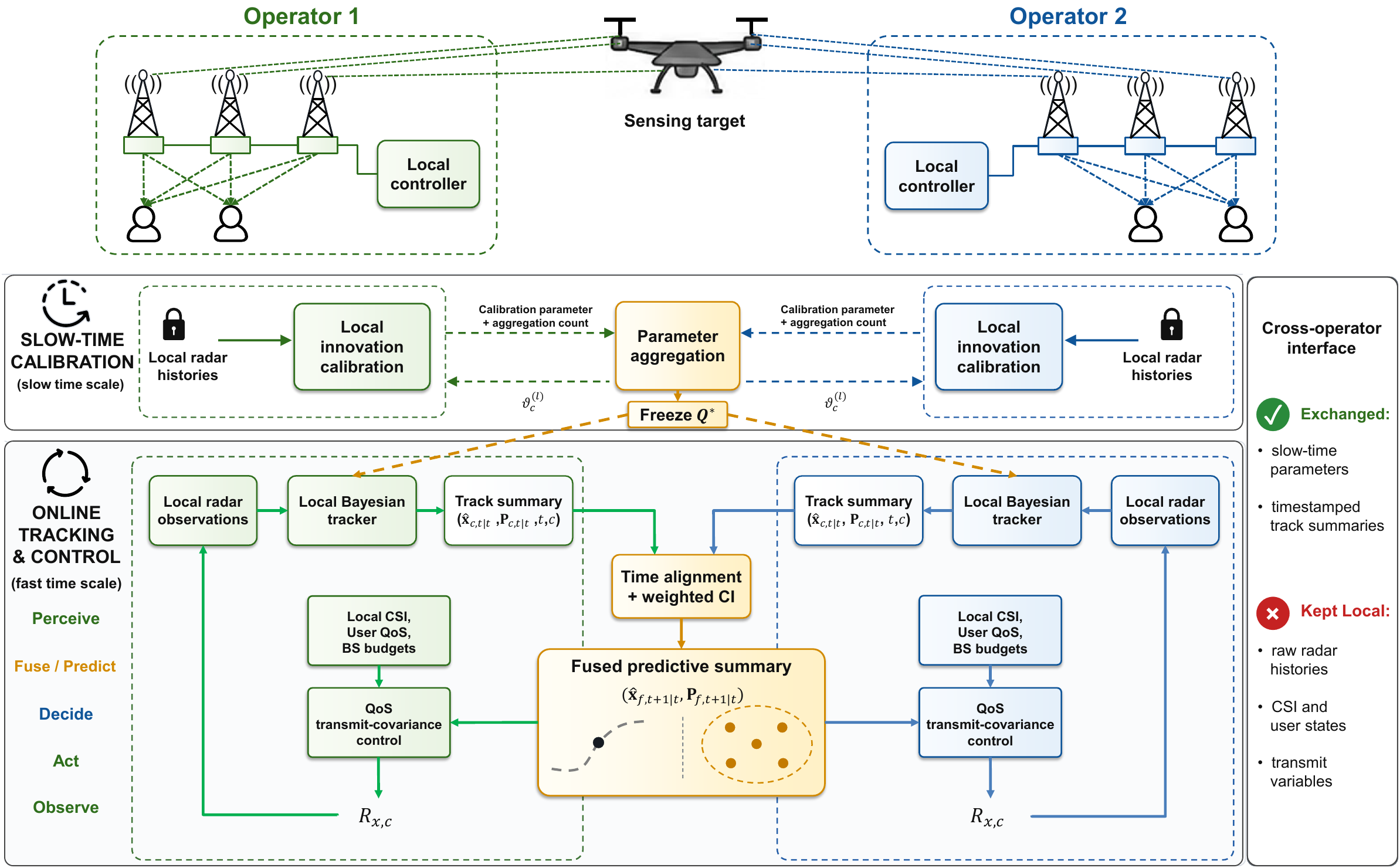}
\caption{Two-timescale architecture for multi-operator ISAC. Calibration parameters are exchanged at the slower time scale, while Gaussian track summaries are exchanged online; radar histories, CSI, user states, and transmission variables remain within each operator.}
\label{fig:system_model_covariance_interface}
\end{figure*}

Fig.~\ref{fig:system_model_covariance_interface} summarizes the two-timescale architecture and the flow of information from sensing and tracking to track fusion and transmission control.

\subsection{Multi-Operator ISAC Network}
\label{subsec:network_model}

Consider a multi-operator ISAC network with operator set $\mathcal C=\{1,\ldots,C\}$. Operator $c\in\mathcal C$ controls a CoMP cluster of BSs $\mathcal B_c=\{1,\ldots,B_c\}$ and serves communication users $\mathcal K_c=\{1,\ldots,K_c\}$. BS $b\in\mathcal B_c$ is located at
\begin{equation}
    \mathbf s_{c,b}
    =
    [x_{c,b},y_{c,b},h_{\rm BS}]^{\mathsf T}.
    \label{eq:bs_location}
\end{equation}
It has $N_t$ transmit antennas, giving the aggregate cluster dimension $M_c=B_cN_t$. A common $N_t$ is used for compact notation. CoMP coordination is confined to the BSs belonging to the same operator and no joint transmission or beamforming is performed across operators. At slot \(t\), operator \(c\) applies a total transmit covariance \(\mathbf R^{\rm tx}_{c,t}\in\mathbb C^{M_c\times M_c}\), \(\mathbf R^{\rm tx}_{c,t}\succeq 0\), which determines both downlink transmission and directional sensing power. Its detailed parameterization is given in Section~\ref{subsec:downlink_covariance_model}.

Operators use orthogonal spectrum resources, and the evaluated model therefore neglects both cross-operator communication interference and cross-operator sensing interference. Cooperation between operators therefore occurs through the information exchange described in Section~\ref{subsec:information_boundary}.

For analytical simplicity, we assume the UAV is a fixed-altitude sensing target. Its horizontal kinematic state and position are
\begin{equation}
    \mathbf x_t
    =
    [x_t,y_t,v_{x,t},v_{y,t}]^{\mathsf T},
    \qquad
    \mathbf q_t=[x_t,y_t]^{\mathsf T}.
    \label{eq:uav_state}
\end{equation}
For radar geometry, path loss, and array response, the horizontal position is embedded in three dimensions as $[\mathbf q_t^{\mathsf T},h_{\rm U}]^{\mathsf T}$, where the altitude $h_{\rm U}$ is known and fixed. Communication QoS is imposed only on the users in $\mathcal K_c$.

The slot convention is causal and uniform throughout the paper. After processing the observations available at slot $t$, the fusion layer forms $(\hat{\mathbf x}_{f,t|t},\mathbf P_{f,t|t})$. Its one-step prediction $(\hat{\mathbf x}_{f,t+1|t},\mathbf P_{f,t+1|t})$ informs the transmit covariance $\mathbf R_{c,t}^{\rm tx}$, and that transmission produces the post-acquisition radar observation at slot $t+1$. Thus, an action indexed by $t$ never uses an observation indexed by $t+1$.

\subsection{Tracking Information Exchange}
\label{subsec:information_boundary}

Under the Gaussian filtering approximation, coordination between operators uses two compact messages:
\begin{itemize}
    \item on the slow time scale, operator $c$ reports a low-dimensional calibration parameter and a scalar count used for aggregation; and
    \item on the slot time scale, operator $c$ reports a timestamped Gaussian track summary $(\hat{\mathbf x}_{c,t|t},\mathbf P_{c,t|t},t,c)$.
\end{itemize}
The fixed target identifier and coordinate-frame information are treated as metadata and omitted from the tuple for compact notation. Raw radar data and innovations, CSI, user information, and transmit variables are not exchanged.

All local trackers are initialized from the same acquisition posterior $(\hat{\mathbf x}_{0|0}^{\rm acq},\mathbf P_{0|0}^{\rm acq},t=0)$. They subsequently process their own sensing observations and evolve independently. Since the local estimates inherit common prior information while tracking the same target, their estimation errors may be correlated across operators. CI is therefore applied to time-aligned Gaussian summaries of the same registered target without requiring the corresponding error cross-covariances.

\subsection{Fixed NCV Motion Model}
\label{subsec:motion_model}

Each local tracker uses a fixed-altitude NCV model from the standard family of maneuvering-target dynamic models \cite{li2003surveyManeuveringTargets}. Maneuvers not represented by the nominal NCV dynamics are captured through the residual process covariance \(\mathbf Q^\star\), which is calibrated in Section~\ref{subsec:federated_motion_learning}. The state transition is
\begin{equation}
    p_{\rm trk}(\mathbf x_{t+1}\mid\mathbf x_t)
    =
    \mathcal N\!\left(
    \mathbf x_{t+1};\mathbf F\mathbf x_t,\mathbf Q^\star
    \right),
    \label{eq:single_motion_model}
\end{equation}
where \(\mathbf F\in\mathbb R^{4\times4}\) and \(\mathbf G\in\mathbb R^{4\times2}\) are
\begin{equation}
    \mathbf F=
    \begin{bmatrix}
        \mathbf I_2 & \Delta t\,\mathbf I_2\\
        \mathbf0 & \mathbf I_2
    \end{bmatrix},
    \qquad
    \mathbf G=
    \begin{bmatrix}
        \frac{\Delta t^2}{2}\mathbf I_2\\
        \Delta t\,\mathbf I_2
    \end{bmatrix}.
    \label{eq:ncv_block_matrices}
\end{equation}
Equivalently, the process perturbation is generated by a slot-wise random acceleration
\begin{equation}
    \mathbf a_t\sim \mathcal N\!\left( \mathbf 0, \operatorname{diag}(q_x^\star,q_y^\star) \right), \qquad \mathbf w_t=\mathbf G\mathbf a_t.
\end{equation}
Hence, the deployed process covariance is
\begin{equation}
    \mathbf Q^\star
    =
    \mathbf G
    \operatorname{diag}(q_x^\star,q_y^\star)
    \mathbf G^{\mathsf T},
    \qquad
    q_x^\star,q_y^\star\geq q_{\min}>0.
    \label{eq:main_constant_process_covariance}
\end{equation}
Here, \(q_x^\star\) and \(q_y^\star\) are the slot-wise horizontal acceleration variances. The isotropic specialization defines $q^\star\triangleq q_x^\star=q_y^\star$. Once calibrated, $\mathbf Q^\star$ remains fixed throughout the online run and is used for both local and fused-state predictions.

\subsection{Radar Observation and Availability Model}
\label{subsec:radar_observation_model}

For $t\geq0$, the applied total covariance $\mathbf R_{c,t}^{\rm tx}$ determines the directional sensing power and hence the precision of the radar observations available in the next slot. If BS \(b\) has valid geometry and an unoccluded line-of-sight (LoS) channel, it reports
\begin{subequations}
\begin{equation}
\mathbf z_{c,b,t+1}
    =
    \mathbf g_{c,b}(\mathbf x_{t+1})
    +
    \mathbf n_{c,b,t+1},
    \label{eq:radar_measurement_model}
\end{equation}
\begin{equation}
    \mathbf g_{c,b}(\mathbf x)
    =
    [r_{c,b}(\mathbf x),\varphi_{c,b}(\mathbf x),\dot r_{c,b}(\mathbf x)]^{\mathsf T}.
\end{equation}
\end{subequations}
The measurement noise is generated as $\mathbf n_{c,b,t+1} \sim \mathcal N \left( \mathbf 0, \mathbf R^{\rm true}_{c,b,t+1} \right),$ where \(\mathbf R^{\rm true}_{c,b,t+1}\in\mathbb R^{3\times3}\) is the covariance associated with the realized geometry and applied transmission. Let $\boldsymbol\delta_{c,b}(\mathbf x)=
[x-x_{c,b},y-y_{c,b},h_{\rm U}-h_{\rm BS}]^{\mathsf T}$. The measurement components are
\begin{subequations}
\begin{align}
    r_{c,b}(\mathbf x)
    &=
    \|\boldsymbol\delta_{c,b}(\mathbf x)\|_2,
    \label{eq:range_model}\\
    \varphi_{c,b}(\mathbf x)
    &=
    \operatorname{atan2}(y-y_{c,b},x-x_{c,b}),
    \label{eq:azimuth_model}\\
    \dot r_{c,b}(\mathbf x)
    &=
    \frac{
    ([x,y]^{\mathsf T}-[x_{c,b},y_{c,b}]^{\mathsf T})^{\mathsf T}
    [v_x,v_y]^{\mathsf T}}
    {r_{c,b}(\mathbf x)}.
    \label{eq:radial_velocity_model}
\end{align}
\end{subequations}

Availability is modeled separately from observation precision. Define
\begin{equation}
    \mathcal V_{c,t+1}
    =
    \bigl\{b\in\mathcal B_c :
    a^{\rm geo}_{c,b,t+1}=1,\;
    a^{\rm occ}_{c,b,t+1}=1,
    p^{\rm dir}_{c,b,t+1}>0\bigr\},
    \label{eq:radar_detection_set}
\end{equation}
where the indicators encode valid geometry and absence of occlusion, and $p^{\rm dir}_{c,b,t+1}$ is induced by $\mathbf R_{c,t}^{\rm tx}$. The set $\mathcal V_{c,t+1}$ identifies usable BSs, whereas $\mathbf R^{\rm rep}_{c,b,t+1}$ quantifies their observation quality. For a usable BS, reduced directional power is represented by increased measurement uncertainty through~\eqref{eq:snr_observation_covariance}. A BS outside the set contributes no observation, and the tracker coasts when the set is empty.

The receiver evaluates
\begin{equation}
    \mathbf R_{c,b,t+1}^{\rm rep}
    =
    \mathcal R_{c,b}^{\rm rep}
    \left(
    \hat{\mathbf x}_{c,t+1|t},
    \mathbf R_{c,t}^{\rm tx},
    \boldsymbol\vartheta_{c,b},
    \mathcal V_{c,t+1}
    \right).
    \label{eq:radar_observation_covariance}
\end{equation}
The mapping uses the predicted target geometry, the applied transmit covariance, and the declared availability state. A diagonal instance based on the predicted signal-to-noise ratio (SNR) is
\begin{equation}
    \mathbf R_{c,b,t+1}^{\rm rep}
    =
    \frac{
    \operatorname{diag}(\alpha_r,\alpha_\varphi,\alpha_{\dot r})
    }
    {\operatorname{SNR}_{c,b,t+1}
    (\hat{\mathbf x}_{c,t+1|t},\mathbf R_{c,t}^{\rm tx})},
    \qquad
    b\in\mathcal V_{c,t+1}.
    \label{eq:snr_observation_covariance}
\end{equation}
This inverse-SNR scaling follows the standard dependence of radar parameter-estimation variance on received SNR~\cite{liu2020jointRadarComm,liu2022crbJrc}. The coefficients \(\alpha_r\), \(\alpha_\phi\), and \(\alpha_{\dot r}\) are the corresponding unit-SNR variance coefficients for range, azimuth, and radial velocity. In simulation, \(\mathbf R^{\rm true}_{c,b,t+1}\) is evaluated at the realized target geometry and is used to generate the measurement noise, whereas tracking and calibration use the receiver-reported covariance \(\mathbf R^{\rm rep}_{c,b,t+1}\). For the available BSs, the reported covariances are stacked into the block-diagonal matrix \(\mathbf R^{\rm rep}_{c,t+1}\).

\subsection{Local EKF Tracking}
\label{subsec:local_bayesian_tracking}

Given the fixed $\mathbf Q^\star$, each operator runs an extended Kalman filter (EKF) using the nonlinear radar observation model in Section~\ref{subsec:radar_observation_model} \cite{bar2001estimation}. From the posterior at slot $t$, the next-slot prior is
\begin{subequations}
\begin{align}
    \hat{\mathbf x}_{c,t+1|t}
    &=
    \mathbf F\hat{\mathbf x}_{c,t|t},
    \label{eq:local_predict_mean}\\
    \mathbf P_{c,t+1|t}
    &=
    \mathbf F\mathbf P_{c,t|t}\mathbf F^{\mathsf T}
    +\mathbf Q^\star.
    \label{eq:local_predict_cov}
\end{align}
\end{subequations}
For a fixed ordering of $\mathcal V_{c,t+1}$, define
\begin{subequations}
    \begin{align}
    \mathbf z_{c,t+1}
    &\triangleq\operatorname{col}_{b\in\mathcal V_{c,t+1}}\mathbf z_{c,b,t+1},\\
    \mathbf g_{c,t+1}(\mathbf x)
    &\triangleq\operatorname{col}_{b\in\mathcal V_{c,t+1}}\mathbf g_{c,b}(\mathbf x),
    \end{align}
\end{subequations}
both in \(\mathbb R^{3|\mathcal V_{c,t+1}|}\). Their Jacobians are stacked into $\mathbf H^{\rm tr}_{c,t+1} \in \mathbb R^{3|\mathcal V_{c,t+1}|\times4},$ while $\mathbf R^{\rm rep}_{c,t+1} \in \mathbb R^{ 3|\mathcal V_{c,t+1}| \times 3|\mathcal V_{c,t+1}| }$ contains the corresponding reported measurement covariances.
\begin{subequations}
\begin{align}
    \boldsymbol\nu_{c,t+1}
    &=\mathbf z_{c,t+1}-\mathbf g_{c,t+1}(\hat{\mathbf x}_{c,t+1|t}),
    \label{eq:stacked_operator_observation}\\
    \boldsymbol\Omega_{c,t+1}
    &=\mathbf H_{c,t+1}^{\rm tr}\mathbf P_{c,t+1|t}
    (\mathbf H_{c,t+1}^{\rm tr})^{\mathsf T}
    +\mathbf R_{c,t+1}^{\rm rep}.
    \label{eq:local_innovation_covariance}
\end{align}
\end{subequations}
Azimuth innovations are wrapped to the principal interval before the EKF update. The stacked innovation is accepted when
\begin{equation}
    \boldsymbol\nu_{c,t+1}^{T} \boldsymbol\Omega_{c,t+1}^{-1} \boldsymbol\nu_{c,t+1} \le \chi^2_{d,\eta}, \qquad d=3|\mathcal V_{c,t+1}|,
\end{equation}
where \(\eta\) is the prescribed gating quantile. The posterior covariance is evaluated in Joseph form. If no usable measurement is available or the innovation is rejected, the predicted state and covariance are carried forward as the posterior.

\subsection{Downlink Transmission Model}
\label{subsec:downlink_covariance_model}

The total transmit covariance introduced in Section~\ref{subsec:network_model} is parameterized as
\begin{equation}
    \mathbf R_{c,t}^{\rm tx}
    =
    \sum_{k\in\mathcal K_c}\mathbf W_{c,k,t}
    +
    \mathbf S_{c,t},
    \label{eq:total_transmit_covariance}
\end{equation}
where \(\mathbf W_{c,k,t}\succeq0\) is the covariance assigned to communication user \(k\), and \(\mathbf S_{c,t}\succeq0\) represents the remaining transmit covariance not assigned to a communication-user stream. Its spatial structure is optimized jointly with the user covariances. \(\mathbf R^{\rm tx}_{c,t}\) determines both the S\&C power allocations.

The communication receiver treats the component represented by \(\mathbf S_{c,t}\) as uncancelled interference. The interference-plus-noise power of user \(k\) is
\begin{equation}
    I_{c,k,t}
    =
    \sum_{\ell\in\mathcal K_c\setminus\{k\}}
    \mathbf h_{c,k,t}^{\mathsf H}
    \mathbf W_{c,\ell,t}
    \mathbf h_{c,k,t}
    +
    \mathbf h_{c,k,t}^{\mathsf H}
    \mathbf S_{c,t}
    \mathbf h_{c,k,t}
    +
    \sigma_{c,k}^2.
    \label{eq:interference_noise_term}
\end{equation}
The downlink SINR and communication QoS constraint are
\begin{equation}
    \operatorname{SINR}_{c,k,t}
    =
    \frac{\mathbf h_{c,k,t}^{\mathsf H}\mathbf W_{c,k,t}\mathbf h_{c,k,t}}
    {I_{c,k,t}}
    \geq\gamma_{c,k},
    \qquad k\in\mathcal K_c.
    \label{eq:hard_sinr_constraint}
\end{equation}
Let $\mathbf E_{c,b}\in\{0,1\}^{N_t\times M_c}$ denote the real binary selection matrix that extracts the antenna block of BS \(b\). Its transmit-power constraint is
\begin{equation}
    \operatorname{tr}
    \left(
    \mathbf E_{c,b}\mathbf R_{c,t}^{\rm tx}\mathbf E_{c,b}^{\mathsf H}
    \right)
    \leq
    P_{c,b}^{\max},
    \qquad b\in\mathcal B_c.
    \label{eq:per_bs_covariance_power}
\end{equation}
Here, \(P^{\max}_{c,b}\) denotes the transmit-power budget of BS \(b\).

\section{Federated Calibration and Track Fusion}
\label{sec:federated_tracking}

\subsection{Process-Covariance Calibration from Radar Innovations}
\label{subsec:federated_motion_learning}

The slow-time stage estimates the fixed process-covariance level from radar innovations recorded by each operator, following residual-based covariance-identification methods \cite{mehra1970identification,odelson2006autocovariance}. Since the experiments use the isotropic model, we parameterize the acceleration variance by a scalar $\theta\in\mathbb R$:
\begin{equation}
    q(\theta)
    =
    q_{\min}+\exp(\theta),
    \qquad
    \mathbf Q(\theta)
    =
    \mathbf G
    \bigl(q(\theta)\mathbf I_2\bigr)
    \mathbf G^{\mathsf T}.
\end{equation}
This parameterization enforces \(q(\theta)>q_{\min}\) while allowing local optimization and aggregation to be performed over the unconstrained parameter \(\theta\).
\footnote{The implementation also supports axis-specific calibration by using \(\boldsymbol\theta=(\theta_x,\theta_y)\) and replacing \(q(\theta)\mathbf I_2\) with \(\operatorname{diag}(q_{\min}+e^{\theta_x},q_{\min}+e^{\theta_y})\).}

Let \(\lambda\ge0\) denote the proximal regularization coefficient, which limits the departure of each local fit from the parameter broadcast at the current aggregation round. Let \(\mathcal T_c^{\rm cal}\) denote the calibration slots scored by operator \(c\). For a candidate \(\theta\), the operator replays the EKF in Section\ref{subsec:local_bayesian_tracking} with \(\mathbf Q^\star\) replaced by \(\mathbf Q(\theta)\), yielding the innovation \(\boldsymbol\nu_{c,t}(\theta)\) and covariance \(\boldsymbol\Omega_{c,t}(\theta)\). Receiver noises are assumed conditionally independent across BSs, so the recorded reported covariances are stacked block diagonally over the valid observations. Calibration uses the same receiver-reported covariance as the tracker.

At aggregation round $\ell$, operator $c$ minimizes
\begin{equation}
\begin{aligned}
    \mathcal L_c^{\rm cal}
    (\theta;\theta^{(\ell)})
    ={}&
    \frac{1}{2N_c}
    \sum_{t\in\mathcal T_c^{\rm cal}}
    \left[
    \log\left|\boldsymbol\Omega_{c,t}\right|
    +
    \boldsymbol\nu_{c,t}^{\mathsf T}
    \boldsymbol\Omega_{c,t}^{-1}
    \boldsymbol\nu_{c,t}
    \right]\\
    &+
    \lambda
    \left(\theta-\theta^{(\ell)}\right)^2.
\end{aligned}
\label{eq:local_innovation_nll}
\end{equation}
Here, $N_c
    =
    \sum_{t\in\mathcal T_c^{\rm cal}}
    |\mathcal V_{c,t}|$
is the number of valid BS--time tuples contributing to the scored innovations.\footnote{The implementation accumulates the corresponding counter over a common number \(E\) of local optimization passes, yielding \(\widetilde N_c=EN_c\). Since the same \(E\) is used by all operators, this factor cancels in the normalized aggregation weights.} The first usable observation of each calibration trajectory is used for initialization and is excluded from the score. Normalization by \(N_c\) expresses the data-fit term per valid BS--time tuple. Terms independent of \(\theta\), including the Gaussian normalization constant, are omitted from~\eqref{eq:local_innovation_nll}.

Starting from the broadcast $\theta^{(\ell)}$, each operator performs a fixed number of local optimization steps and returns $\theta_c^{(\ell+1)}$ together with its aggregation count. Let
\begin{equation}
    \mathcal C_+=\{c\in\mathcal C:N_c>0\},
    \qquad \sum_{c\in\mathcal C_+}N_c>0.
\end{equation}

The server applies count-weighted aggregation,
\begin{equation}
    \theta^{(\ell+1)}
    =
    \sum_{c\in\mathcal C_+}
    \alpha_c\theta_c^{(\ell+1)},
    \qquad
    \alpha_c
    =
    \frac{N_c}{\sum_{j\in\mathcal C_+}N_j}.
    \label{eq:federated_motion_aggregation}
\end{equation}
Aggregation in the log-excess-variance domain corresponds to a weighted geometric average of \(q-q_{\min}\). Equation~\eqref{eq:federated_motion_aggregation} follows the data-size-weighted aggregation principle of FedAvg \cite{mcmahan2017fedAvg}, applied here to the low-dimensional calibration parameter. The final aggregate \(\theta^{\rm agg}\) defines \(\mathbf Q^\star=\mathbf Q(\theta^{\rm agg})\), which is fixed during online tracking.

\subsection{Weighted CI of Track Summaries}
\label{subsec:track_fusion}

\subsubsection{Time Alignment and Fusion Rule}
\label{subsec:fusion_message}

Let $\mathcal A_t\subseteq\mathcal C$ denote the active operators with valid, time-aligned summaries for the registered target at slot $t$. For \(c\in\mathcal A_t\), \(\mathbf P_{c,t|t}\succ0\), and define
\begin{equation}
    \boldsymbol\Lambda_{c,t}
    =
    \mathbf P_{c,t|t}^{-1},
    \qquad
    \boldsymbol\eta_{c,t}
    =
    \boldsymbol\Lambda_{c,t}\hat{\mathbf x}_{c,t|t}.
    \label{eq:information_form_track}
\end{equation}
For weights in the closed simplex
\begin{equation}
    \mathcal W_t
    =
    \left\{
    \boldsymbol\omega:
    \omega_c\geq0,\;
    \sum_{c\in\mathcal A_t}\omega_c=1
    \right\},
    \label{eq:fusion_simplex}
\end{equation}
weighted CI forms \cite{julier1997nonDivergent,chen2002ciRevisited}
\begin{equation}
\label{eq:CI_forms}
\begin{split}
    \boldsymbol\Lambda_{f,t}(\boldsymbol\omega)
    &=
    \sum_{c\in\mathcal A_t}\omega_c\boldsymbol\Lambda_{c,t},\quad
    \boldsymbol\eta_{f,t}(\boldsymbol\omega)
    =
    \sum_{c\in\mathcal A_t}\omega_c\boldsymbol\eta_{c,t},\\
    \mathbf P_{f,t|t}(\boldsymbol\omega)
    &=
    [\boldsymbol\Lambda_{f,t}(\boldsymbol\omega)]^{-1},\quad
    \hat{\mathbf x}_{f,t|t}(\boldsymbol\omega)
    =
    \mathbf P_{f,t|t}(\boldsymbol\omega)
    \boldsymbol\eta_{f,t}(\boldsymbol\omega).
\end{split}
\end{equation}
When the contributing local covariances are consistent, CI preserves covariance consistency without requiring their error cross-correlations to be known~\cite{julier1997nonDivergent,chen2002ciRevisited}.
\subsubsection{Trace-Based CI Weight Selection}
\label{subsec:ci_weight_selection}

The CI weights are selected by minimizing the trace of the fused state covariance, which provides a scalar measure of the aggregate marginal uncertainty:
\begin{equation}
    \hat{\boldsymbol\omega}_t
    \in
    \underset{\substack{\boldsymbol\omega\in\mathcal W_t}}
    {\operatorname{arg\,min}}
    \operatorname{tr}
    \left[
    \left(
    \sum_{c\in\mathcal A_t}
    \omega_c\mathbf P_{c,t|t}^{-1}
    \right)^{-1}
    \right].
    \label{eq:grid_ci_weight}
\end{equation}
For the two-operator setting considered in the experiments, the simplex reduces to \(\omega_1\in[0,1]\), with \(\omega_2=1-\omega_1\), and~\eqref{eq:grid_ci_weight} is evaluated on a uniform 101-point grid. For more than two active operators, the weight search extends to the corresponding multidimensional simplex.

The selected weights are substituted into~\eqref{eq:CI_forms}, giving
\begin{equation}
    (\hat{\mathbf x}_{f,t|t},\mathbf P_{f,t|t})
    =
    \left(
    \hat{\mathbf x}_{f,t|t}(\hat{\boldsymbol\omega}_t),
    \mathbf P_{f,t|t}(\hat{\boldsymbol\omega}_t)
    \right).
    \label{eq:final_fused_track}
\end{equation}

\begin{lemma}[CI-induced covariance bound]
\label{lem:ci_prior_comparability}
For any $\boldsymbol\omega\in\mathcal W_t$ and a reference operator $c\in\mathcal A_t$, define
\begin{subequations}
\begin{align}
    \kappa_{j|c,t}
    &=\lambda_{\max}\!\left(
    \boldsymbol\Lambda_{c,t}^{-1/2}
    \boldsymbol\Lambda_{j,t}
    \boldsymbol\Lambda_{c,t}^{-1/2}
    \right),\\
    \rho_{c,t}
    &=\max\!\left\{1,\sum_{j\in\mathcal A_t}
    \omega_j\kappa_{j|c,t}\right\}.
\end{align}
\label{eq:ci_prior_comparability_factor}
\end{subequations}
Then
\begin{equation}
    \mathbf P_{c,t|t}
    \preceq
    \rho_{c,t}\mathbf P_{f,t|t}(\boldsymbol\omega).
    \label{eq:ci_posterior_comparability}
\end{equation}
When the local and fused summaries are propagated using the same \(\mathbf F\) and \(\mathbf Q^\star\), the same factor also bounds their one-step predictive covariances and the corresponding position blocks.
\end{lemma}
The proof is given in Appendix~\ref{app:proof_ci_prior_comparability}.

\subsection{One-Step Prediction and Design Points}
\label{subsec:predictive_interface}

The fused posterior is propagated once to the slot for which transmission is designed:
\begin{subequations}
\label{eq:fused_prediction}
\begin{align}
    \hat{\mathbf x}_{f,t+1|t}
    &=
    \mathbf F\hat{\mathbf x}_{f,t|t},
    \label{eq:fused_prediction_mean}\\
    \mathbf P_{f,t+1|t}
    &=
    \mathbf F\mathbf P_{f,t|t}\mathbf F^{\mathsf T}
    +
    \mathbf Q^\star.
    \label{eq:fused_prediction_cov}
\end{align}
\end{subequations}
Let \(\hat{\mathbf q}_{f,t+1|t}\) and \(\mathbf P^{\rm pos}_{f,t+1|t}\) denote the horizontal position mean and covariance of the prediction in~\eqref{eq:fused_prediction}. The predictive mean specifies the nominal sensing geometry, while the predictive mean and covariance are used to construct the sensing design points below.

To keep the number of sensing constraints fixed and small while accounting for the two-dimensional position uncertainty, we use a symmetric five-point construction consisting of the predictive mean and two pairs of covariance-dependent perturbations. Let \(\mathbf L_{f,t+1|t}\) be the lower-triangular Cholesky factor satisfying $\mathbf P^{\rm pos}_{f,t+1|t}=\mathbf L_{f,t+1|t}\mathbf L_{f,t+1|t}^{\mathsf T}$. The design points are
\begin{equation}
    \mathbf q_{t+1}^{(i)}
    =
    \hat{\mathbf q}_{f,t+1|t}
    +
    r_{\rm d}\,
    \mathbf L_{f,t+1|t}\boldsymbol\xi^{(i)},
    \quad
    \boldsymbol\xi^{(i)}
    \in
    \{\mathbf0,\,\pm\mathbf e_1,\,\pm\mathbf e_2\}.
    \label{eq:sigma_design_points}
\end{equation}
The vectors \(\mathbf e_1\) and \(\mathbf e_2\) are the canonical basis vectors of \(\mathbb R^2\). The central point evaluates the nominal predicted position, while the four displaced points account for uncertainty along the two directions defined by the Cholesky factor. Each displaced point has squared Mahalanobis distance \(r_d^2\) from the predictive mean and therefore lies on the same covariance contour. The scalar \(r_d\) controls the spread of these points and is treated as a design parameter; the experiments use \(r_d=\sqrt2\). The lower-triangular Cholesky factor fixes the orientation of the construction in the adopted horizontal coordinate frame.

\section{Uncertainty-Aware Transmission Design}
\label{sec:qos_isac_bf}

This section uses the fused predictive state and covariance to formulate the transmit-covariance control problem solved by each operator. The shared prediction is combined with its own CSI, QoS requirements, and power constraints.

\subsection{Affine Position-Information Model}
\label{subsec:sensing_information_model}

Let $\mathbf P_{f,t+1|t}^{\rm pos}\succ\mathbf0$ denote the horizontal position block of the fused prediction, and let \(\{\mathbf q^{(i)}_{t+1}\}_{i=1}^{N_s}\) be the design points defined in~\eqref{eq:sigma_design_points}.

For any horizontal position \(\mathbf q\), let $\mathbf a_{c,b}(\mathbf q)\in \mathbb{C}^{N_t}$, $\|\mathbf a_{c,b}(\mathbf q)\|_2=1$, denote the normalized ULA steering vector of BS \(b\) toward \(\mathbf q\). Define
$\mathbf D_{c,b}(\mathbf q)
=
\mathbf E_{c,b}^{\mathsf H}
\mathbf a_{c,b}(\mathbf q)\mathbf a_{c,b}^{\mathsf H}(\mathbf q)
\mathbf E_{c,b}$.
The directional power delivered by BS $b$ is then
\begin{equation}
    p_{c,b}(\mathbf q;\mathbf R_{c,t}^{\rm tx})
    =
    \operatorname{tr}\!\left(
    \mathbf D_{c,b}(\mathbf q)\mathbf R_{c,t}^{\rm tx}
    \right).
    \label{eq:directional_power_map}
\end{equation}
The fixed-altitude range--azimuth Jacobian is
\begin{equation}
    \mathbf H_{c,b}(\mathbf q)
    =
    \frac{\partial
    [r_{c,b}(\mathbf q),\varphi_{c,b}(\mathbf q)]^{\mathsf T}}
    {\partial\mathbf q^{\mathsf T}}.
    \label{eq:position_observation_jacobian}
\end{equation}
Consistent with the inverse-SNR covariance model in~\eqref{eq:snr_observation_covariance}, the range–azimuth covariance at a fixed geometry scales inversely with directional power. Let $\mathbf R_{c,b}^{\rm ref}(\mathbf q)$denote the declared range–azimuth covariance at a reference directional power \(P^{\rm ref}_{c,b}\). For directional power \(p>0\), $\mathbf R_{c,b}(\mathbf q,p) = \frac{P^{\rm ref}_{c,b}}{p} \mathbf R^{\rm ref}_{c,b}(\mathbf q)$.
The corresponding per-unit-power position information is therefore
\begin{equation}
    \bar{\mathbf J}_{c,b}(\mathbf q)
    =
    \frac{1}{P_{c,b}^{\rm ref}}
    \mathbf H_{c,b}^{\mathsf T}(\mathbf q)
    [\mathbf R_{c,b}^{\rm ref}(\mathbf q)]^{-1}
    \mathbf H_{c,b}(\mathbf q)
    \succeq\mathbf0.
    \label{eq:unit_position_information}
\end{equation}
Summing over the BSs predicted to have line of sight to $\mathbf q$, the design position information is
\begin{equation}
    \mathbf J_c^{\rm des}(\mathbf q;\mathbf R_{c,t}^{\rm tx})
    =
    \sum_{b}
    p_{c,b}(\mathbf q;\mathbf R_{c,t}^{\rm tx})
    \bar{\mathbf J}_{c,b}(\mathbf q),
    \label{eq:sensing_fim}
\end{equation}
where the summation is over the BSs predicted to have line of sight to \(\mathbf q\). For fixed design points, predicted availability, and reference parameters, \(\mathbf J^{\rm des}_{c}\) is a positive-semidefinite (PSD) affine function of \(\mathbf R^{\rm tx}_{c,t}\).

At design point $i$, define
\begin{equation}
    \boldsymbol\Sigma_{c,t,i}(\mathbf R_{c,t}^{\rm tx})
    =
    \left[
    (\mathbf P_{f,t+1|t}^{\rm pos})^{-1}
    +
    \mathbf J_c^{\rm des}
    (\mathbf q_{t+1}^{(i)};\mathbf R_{c,t}^{\rm tx})
    \right]^{-1}.
    \label{eq:posterior_position_cov_after_sensing}
\end{equation}
The transmission design imposes
\begin{equation}
    \operatorname{tr}\!\left(
    \boldsymbol\Sigma_{c,t,i}(\mathbf R_{c,t}^{\rm tx})
    \right)
    \leq
    \epsilon,
    \qquad i=1,\ldots,N_s.
    \label{eq:tracking_information_constraint}
\end{equation}
The trace measures the total horizontal position variance, so \(\epsilon\) limits the aggregate predicted position uncertainty at each design point. The predictive covariance enters both the prior-information term in~\eqref{eq:posterior_position_cov_after_sensing} and the construction of the design points in~\eqref{eq:sigma_design_points}.

\subsection{Problem Formulation and QoS Constraints}
\label{subsec:main_bf_problem}

For operator \(c\) at slot \(t\), use the covariance decomposition in~\eqref{eq:total_transmit_covariance}. The transmission design is
\begin{subequations}
\label{prob:qos_isac_covariance}
\begin{align}
    \underset{\{\mathbf W_{c,k,t}\},\,\mathbf S_{c,t}}
    {\min}
    \quad
    &
    \operatorname{tr}(\mathbf R_{c,t}^{\rm tx})
    \label{prob:bf_objective}\\
    \text{s.t.}
    \quad
    &
    \mathbf W_{c,k,t}\succeq\mathbf0,
    \quad
    \mathbf S_{c,t}\succeq\mathbf0,
    \quad k\in\mathcal K_c,
    \label{prob:psd_constraints}\\
    &
    \mathbf h_{c,k,t}^{\mathsf H}
    \mathbf W_{c,k,t}
    \mathbf h_{c,k,t}
    \geq
    \gamma_{c,k} I_{c,k,t},
    \quad k\in\mathcal K_c,
    \label{prob:hard_qos_constraints}\\
    &
    \operatorname{tr}
    \left(
    \mathbf E_{c,b}\mathbf R_{c,t}^{\rm tx}\mathbf E_{c,b}^{\mathsf H}
    \right)
    \leq
    P_{c,b}^{\max},
    \quad b\in\mathcal B_c,
    \label{prob:per_bs_constraints}\\
    &
    \operatorname{tr}
    \left(
    \boldsymbol\Sigma_{c,t,i}(\mathbf R_{c,t}^{\rm tx})
    \right)
    \leq
    \epsilon,
    \quad i=1,\ldots,N_s.
    \label{prob:tracking_constraints}
\end{align}
\end{subequations}
The SINR constraints become affine after rearrangement, while the per-BS constraints impose finite transmit-power budgets. Depending on the QoS targets, tracking threshold, and available power,~\eqref{prob:qos_isac_covariance} may be infeasible; this case is handled explicitly in Algorithm~\ref{alg:closed_loop_isac_tracking}.

\subsection{Semidefinite Reformulation and Rank-One Recovery}
\label{subsec:sdp_interpretation}

Introduce a real symmetric auxiliary matrix \(\mathbf U_{c,t,i}\in\mathbb R^{2\times2}\) for each design point. The tracking constraint in~\eqref{eq:tracking_information_constraint} can be written as
\begin{subequations}
\begin{align}
    &
    \begin{bmatrix}
        [\boldsymbol\Sigma_{c,t,i}
        (\mathbf R_{c,t}^{\rm tx})]^{-1}
        & \mathbf I_2\\
        \mathbf I_2 & \mathbf U_{c,t,i}
    \end{bmatrix}
    \succeq\mathbf0,
    \label{eq:block_SD}\\
    &
    \operatorname{tr}(\mathbf U_{c,t,i})
    \leq
    \epsilon.
\end{align}
    \label{eq:schur_tracking_constraint}
\end{subequations}
Since \(\mathbf P^{\rm pos}_{f,t+1|t}\succ0\) and
\(\mathbf J^{\rm des}_{c}(\mathbf q^{(i)}_{t+1}; \mathbf R^{\rm tx}_{c,t})\succeq0\), the upper-left block in~\eqref{eq:block_SD} is positive definite. By the Schur complement,~\eqref{eq:block_SD} is equivalent to
\begin{equation}
    \mathbf U_{c,t,i} \succeq \boldsymbol\Sigma_{c,t,i}(\mathbf R^{\rm tx}_{c,t}).
\end{equation}
Hence,~\eqref{eq:schur_tracking_constraint} are equivalent to the trace constraint in~\eqref{eq:tracking_information_constraint}: if~\eqref{eq:tracking_information_constraint} holds, one may choose \(\mathbf U_{c,t,i}=\boldsymbol\Sigma_{c,t,i}\), while any feasible \(\mathbf U_{c,t,i}\) in~\eqref{eq:schur_tracking_constraint} implies
\(\operatorname{tr}(\boldsymbol\Sigma_{c,t,i})\le\epsilon\).

Because the upper-left block is affine in \(\mathbf R^{\rm tx}_{c,t}\), each design point contributes a \(4\times4\) linear matrix inequality (LMI). Together with the PSD, SINR, and per-BS power constraints, problem~\eqref{prob:qos_isac_covariance} is therefore an SDP in the covariance variables~\cite{boyd2004convex}.

General rank-reduction results for beamforming SDPs are available~\cite{huang2010rankConstrained}. For the present formulation, Proposition~\ref{prop:rank_one_recovery} gives a constructive recovery by transferring excess rank in a user covariance to \(\mathbf S_{c,t}\) while preserving the total transmit covariance and desired received power.

\begin{proposition}[Exact rank-one recovery of communication covariances]
\label{prop:rank_one_recovery}
Consider a feasible solution $\{\mathbf W_{c,k,t}\}_{k\in\mathcal K_c},
\mathbf S_{c,t}$ of \eqref{prob:qos_isac_covariance}. Assume
$\gamma_{c,k}>0$ and $\sigma_{c,k}^2>0$ for every communication user. We further assume that, apart from the individual PSD constraints and user SINR constraints, the objective and all remaining constraints depend on
\(\{\mathbf W_{c,k,t}\},\mathbf S_{c,t}\)
only through the total covariance
\(\mathbf R^{\rm tx}_{c,t}\). 

For each user define
\begin{equation}
\begin{aligned}
    d_{c,k,t}
    =\mathbf h_{c,k,t}^{\mathsf H}\mathbf W_{c,k,t}\mathbf h_{c,k,t},\qquad \qquad
    \\
    \bar{\mathbf w}_{c,k,t}
    =\frac{\mathbf W_{c,k,t}\mathbf h_{c,k,t}}
    {\sqrt{d_{c,k,t}}},\qquad
    \bar{\mathbf W}_{c,k,t}
    =\bar{\mathbf w}_{c,k,t}\bar{\mathbf w}_{c,k,t}^{\mathsf H},
\end{aligned}
    \label{eq:rank_one_recovery_definition}
\end{equation}
and 
\begin{equation}
    \boldsymbol\Delta_{c,k,t}=\mathbf W_{c,k,t}-\bar{\mathbf W}_{c,k,t}
    \qquad
    \bar{\mathbf S}_{c,t}
    =\mathbf S_{c,t}+\sum_{k\in\mathcal K_c}\boldsymbol\Delta_{c,k,t}.
    \label{eq:rank_one_residual_definition}
\end{equation}
Then $d_{c,k,t}>0$, $\boldsymbol\Delta_{c,k,t}\succeq\mathbf0$,
$\bar{\mathbf S}_{c,t}\succeq\mathbf0$, and
$\operatorname{rank}(\bar{\mathbf W}_{c,k,t})=1$. Moreover,
\begin{equation}
    \sum_k\bar{\mathbf W}_{c,k,t}+\bar{\mathbf S}_{c,t}
    =\sum_k\mathbf W_{c,k,t}+\mathbf S_{c,t}
    =\mathbf R_{c,t}^{\rm tx}.
    \label{eq:rank_one_total_covariance}
\end{equation}
The transformed solution preserves the objective, all user SINRs, the per-BS power constraints, and the tracking-information constraints in~\eqref{prob:qos_isac_covariance}. Equations~\eqref{eq:rank_one_recovery_definition}--\eqref{eq:rank_one_residual_definition} therefore provide an exact rank-one recovery of the communication covariances.
\end{proposition}
The proof is given in Appendix~\ref{app:proof_rank_one_recovery}.

The recovery relies on \(\mathbf S_{c,t}\) being a free PSD variable and on the objective and the remaining non-SINR constraints depending on the covariance variables through their sum \(\mathbf R^{\rm tx}_{c,t}\). Since proposition~\ref{prop:rank_one_recovery} preserves this total covariance, its decomposition into the user covariances and \(\mathbf S_{c,t}\) is generally not unique. We therefore evaluate transmission performance using the total transmit power, constraint satisfaction, and the resulting tracking metrics.

After solving~\eqref{prob:qos_isac_covariance}, finite covariance solutions are numerically symmetrized and projected onto the PSD cone before application. Algorithm~\ref{alg:closed_loop_isac_tracking} summarizes the resulting closed-loop procedure and the associated numerical post-checks.
\begin{algorithm}[!t]
\caption{Closed-loop Tracking and Transmission Procedure}
\label{alg:closed_loop_isac_tracking}
\begin{algorithmic}[1]
\REQUIRE Fixed $\mathbf Q^\star$, shared acquisition posterior at $t=0$, local trackers, CI rule, and QoS constraints.
\STATE Without applying CI, initialize $\mathbf R_{c,0}^{\rm tx}$ by applying the same one-step predict--solve--project policy to the shared acquisition posterior.
\FOR{$t=1,\ldots,T$}
    \STATE Update local tracks; exchange synchronized summaries and apply CI.
    \STATE Predict the fused summary, form the five design points, and solve \eqref{prob:qos_isac_covariance}.
    \IF{no finite covariance solution is returned}
        \STATE Record infeasibility and coast the next radar update.
    \ELSE
        \STATE Symmetrize/project the covariance matrices and evaluate the SINR, per-BS power, and tracking post-checks.
        \STATE Record the numerical certification status.
        \STATE Apply the finite projected \(\mathbf R^{\rm tx}_{c,t}\); its directional power determines the next-slot observation quality.
    \ENDIF
\ENDFOR
\end{algorithmic}
\end{algorithm}

\subsection{Finite-Interval Covariance Analysis}
\label{subsec:covariance_containment}

The analysis proceeds in three steps. First, the transmission constraint bounds the position covariance predicted by the design model. Second, two covariance conditions transfer this bound to the posterior covariance of a local tracker. Third, the NCV recursion propagates the position bound to the velocity, full-state, and predictive covariances. Define
\begin{equation}
    \mathbf B_{\rm pos}=[\mathbf I_2\;\mathbf0_{2\times2}],
    \qquad
    \mathbf B_{\rm vel}=[\mathbf0_{2\times2}\;\mathbf I_2],
\end{equation}
which select the horizontal position and velocity blocks, respectively. Let $\mathcal I_T=\{t_0,\ldots,T\}$ be a consecutive interval over which the conditions below hold, and fix a reference operator whose track summary is available throughout the interval,
\begin{equation}
    c_0\in\bigcap_{t\in\mathcal I_T}\mathcal A_t.
\end{equation}
At slot \(t\), define the mean-point design information and the local position covariances as
\begin{subequations}
\begin{align}
    &\mathbf J_{c,t}^{\rm des}
    =
    \mathbf J_c^{\rm des}
    (\hat{\mathbf q}_{f,t|t-1};\mathbf R_{c,t-1}^{\rm tx}),
    \label{eq:mean_point_design_information}\\
    &\mathbf P_{c,t|s}^{\rm pos}
    =
    \mathbf B_{\rm pos}\mathbf P_{c,t|s}\mathbf B_{\rm pos}^{\mathsf T},
    \quad s\in\{t-1,t\}.
    \label{eq:local_position_blocks}
\end{align}
\end{subequations}
For positive-definite prior and posterior position covariances satisfying $\mathbf 0 \prec \mathbf P^{\rm pos}_{c,t|t} \preceq \mathbf P^{\rm pos}_{c,t|t-1}$, define the effective receiver-side position information as
\begin{equation}
    \mathbf J_{c,t}^{\rm rx}
    =
    (\mathbf P_{c,t|t}^{\rm pos})^{-1}
    -
    (\mathbf P_{c,t|t-1}^{\rm pos})^{-1}
    \succeq\mathbf0,
    \label{eq:receiver_information_matrix}
\end{equation}
so that $\mathbf P_{c,t|t}^{\rm pos}=[(\mathbf P_{c,t|t-1}^{\rm pos})^{-1}+\mathbf J_{c,t}^{\rm rx}]^{-1}$.

A coasted or gate-rejected update has $\mathbf J_{c,t}^{\rm rx}=\mathbf0$ and remains in $\mathcal I_T$ only if the following premises still hold.

\emph{Step 1---Design bound (A1):}
Since the predictive mean is one of the design points in~\eqref{eq:sigma_design_points}, the tracking constraint at this point gives
\begin{equation}
    \operatorname{tr}
    \left[
    \left(
    (\mathbf P_{f,t|t-1}^{\rm pos})^{-1}
    +\mathbf J_{c_0,t}^{\rm des}
    \right)^{-1}
    \right]
    \leq\epsilon.
    \label{eq:step1}
\end{equation}
For an applied covariance that passes the tracking post-check,~\eqref{eq:step1} is verified directly.

\emph{Step 2---Connection to the local tracker.} The first condition compares the position information realized by the receiver with that used in the transmission design. Define
\begin{equation}
    \beta_t \triangleq \max \left\{ \beta\in[0,1]: \mathbf J^{\rm rx}_{c_0,t} \succeq \beta \mathbf J^{\rm des}_{c_0,t} \right\}.
\end{equation}
Condition (A2) requires \(\beta_t>0\). The factor \(\beta_t\) therefore quantifies the fraction of the design information supported by the realized tracking update. Slots with \(\beta_t=0\) are not included in \(\mathcal I_T\).

The second condition (A3) compares the prior covariance of the reference tracker with the fused predictive covariance:
\begin{equation}
    \mathbf P^{\rm pos}_{c_0,t|t-1} \preceq \rho_t \mathbf P^{\rm pos}_{f,t|t-1}, \qquad \rho_t\ge 1.
\end{equation}
At the first slot of the interval, a valid factor can be computed directly from the two prior covariances as
\begin{equation}
\rho_{t_0}
=
\max\!\left\{
1,\,
\lambda_{\max}\!\left(
\mathbf P^{\rm pos}_{c_0,t_0|t_0-1},
\mathbf P^{\rm pos}_{f,t_0|t_0-1}
\right)
\right\},
\end{equation}
where $\lambda_{\max}(\mathbf A,\mathbf B)$ denotes the
largest generalized eigenvalue of the matrix pair
$(\mathbf A,\mathbf B)$.
For \(t>t_0\), Lemma~\ref{lem:ci_prior_comparability} applied at slot \(t-1\), followed by propagation with the common \(\mathbf F\) and \(\mathbf Q^\star\), provides the valid choice $\rho_t=\rho_{c_0,t-1}$.

Conditions (A2) and (A3) are evaluated from the realized tracking and fusion covariances over the interval. Together with (A1), they determine the intervals for which the following covariance bound applies. Define
\begin{equation}
   \bar{\epsilon} \triangleq \epsilon \max_{t\in\mathcal I_T} \max \left\{ \beta_t^{-1}, \rho_t \right\}. 
\end{equation}
Then (A1)–(A3) imply
\begin{equation}
    \operatorname{tr} \left( \mathbf P^{\rm pos}_{c_0,t|t} \right) \le \bar{\epsilon}, \qquad t\in\mathcal I_T.
    \label{eq:constraint_receiver_bridge}
\end{equation}
The proof is given in Appendix~\ref{app:proof_constraint_receiver_bridge}.

\emph{Step 3---NCV propagation.} Equation~\eqref{eq:constraint_receiver_bridge} bounds the posterior position covariance throughout the interval. The following result propagates this bound through the NCV dynamics to the remaining state components.

\begin{theorem}[Finite-interval local-covariance bound]
\label{thm:covariance_containment}
Let the reference tracker $c_0$ use the NCV prediction in \eqref{eq:local_predict_cov}, with $\Delta t>0$ and $\mathbf P_{c_0,t_0|t_0}\succ\mathbf0$. Assume every update over $\mathcal I_T$ is covariance reducing, $\mathbf0\prec\mathbf P_{c_0,t|t}\preceq\mathbf P_{c_0,t|t-1}$,
where equality represents a coasted update, and suppose that
\(
\operatorname{tr}(\mathbf B_{\rm pos}
\mathbf P_{c_0,t|t}\mathbf B_{\rm pos}^{\mathsf T})
\leq\bar\epsilon
\)
for all $t\in\mathcal I_T$. Then, for every $t_0+1\leq t\leq T$,
\begin{equation}
\begin{aligned}
    \bar V
    \triangleq{\left(
    \frac{2\sqrt{\bar\epsilon}}{\Delta t}
    +\frac{\Delta t}{2}\sqrt{q_x^{\star}+q_y^{\star}}
    \right)}^2,\qquad\\
    \operatorname{tr}
    (\mathbf B_{\rm vel}\mathbf P_{c_0,t|t}\mathbf B_{\rm vel}^{\mathsf T})
    \leq\bar V,\qquad
    \operatorname{tr}(\mathbf P_{c_0,t|t})
    \leq\bar\epsilon+\bar V,    
\end{aligned}
\end{equation}
and, for \(t_0+1\le t<T\),
\begin{equation}
\begin{split}
    \operatorname{tr}(\mathbf P_{c_0,t+1|t})
    \leq&
    {\left(\sqrt{\bar\epsilon}+\Delta t\sqrt{\bar V}\right)}^2
    +\bar V \\
    &+
    \left(\Delta t^2+\frac{\Delta t^4}{4}\right)
    (q_x^{\star}+q_y^{\star}).
\end{split}
\label{eq:covariance_containment_bounds}
\end{equation}
All dependence on the design and fusion stages enters the bounds through \(\bar{\epsilon}\). Once \(\bar{\epsilon}\) is fixed, the remaining terms depend only on the NCV parameters and the sampling interval.
\end{theorem}

The proof is given in Appendix~\ref{app:proof_covariance_containment}.

The term
\(2\sqrt{\bar{\epsilon}}/\Delta t\) in \(\sqrt{\bar V}\) arises from bounding velocity through consecutive position-covariance bounds, while
\((\Delta t/2)\sqrt{q_x^\star+q_y^\star}\)
accounts for the one-slot process-noise contribution.
Theorem~\ref{thm:covariance_containment} characterizes the covariance recursion of the reference tracker over intervals satisfying (A1)–(A3); realized estimation accuracy is evaluated separately in Section~\ref{sec:simulation_results}.

\section{Simulation Results and Discussion}
\label{sec:simulation_results}

\subsection{Simulation Setup and Statistical Protocol}
\label{subsec:eval_protocol}

Evaluation seeds are the statistical units; all slot-, operator-, and trajectory-level observations are aggregated within each seed before paired inference. Paired arms share the trajectory, channel, availability, and measurement-noise realizations~\cite{rubinstein1985commonRandom}. We report normalized estimation error squared (NEES), normalized innovation squared (NIS), 95\% state coverage, root-mean-square error (RMSE), normalized total power, and numerical feasibility checks. NEES and NIS are normalized by their respective dimensions, \(4\) and \(3|\mathcal V|\). We use $0.75\le {\rm NEES}/4\le1.25$ and $0.93\le {\rm coverage}\le0.97$ as predeclared descriptive bands around the nominal values. These bands are used for summary and graphical comparisons rather than as formal statistical acceptance regions. Statistical tests are applied to paired seed-level summaries, with the specific test reported for each inferential comparison.

All online methods start from the common acquisition posterior described in Section~\ref{subsec:information_boundary}, and evaluation begins with the first subsequent sensing update. Table~\ref{tab:simulation_parameters} summarizes the network, radio, calibration, tracking and fusion, QoS control, and evaluation settings. Following Algorithm~\ref{alg:closed_loop_isac_tracking}, all finite projected solutions are applied and retained in the continuous metrics; post-checks provide numerical certification and do not determine sample inclusion. The tracking threshold \(\epsilon=34~{\rm m}^2\) was selected from \(\{16,20,24,28,34,42\}~{\rm m}^2\) using development seeds 0–2. The 0.999 NIS gate was fixed during a separate implementation audit involving seed 9, after which all reported comparison arms were rerun under the same frozen configuration.

\begin{table}[!t]
\centering
\caption{Main simulation and evaluation settings.}
\label{tab:simulation_parameters}
\scriptsize
\setlength{\tabcolsep}{2pt}
\begin{tabular}{@{}>{\raggedright\arraybackslash}p{0.26\columnwidth}>{\raggedright\arraybackslash}p{0.68\columnwidth}@{}}
\hline
Block & Main setting\\
\hline
Network and geometry
& $C=2$, $B_c=4$, $K_c=2$, $N_t=8$; $h_{\rm U}=50$~m,
$h_{\rm BS}=10$~m, and $\Delta t=0.1$~s.\\
Frequency
& 28~GHz / 50~MHz.\\
Calibration
& 8 training / 2 held-out trajectories; Adam, $\text{lr}=0.05$, 150 local epochs, $\lambda = 10^{-4}$\\
Tracking and fusion
& NIS gate at 0.999 chi-square quantile; trace-CI with 101-point grid.\\
QoS control
& $5$-dB SINR floor; $P^{\max}=100$ normalized units per BS; five design points
with $r_{\rm d}=\sqrt{2}$; $\epsilon=34~\mathrm{m^2}$.\\
Evaluation
& 2 held-out trajectories and 11 post-acquisition slots per trajectory;
10 evaluation seeds per regime.\\
\hline
\end{tabular}
\end{table}

\begin{figure}[!t]
\centering
\includegraphics[width=\linewidth]{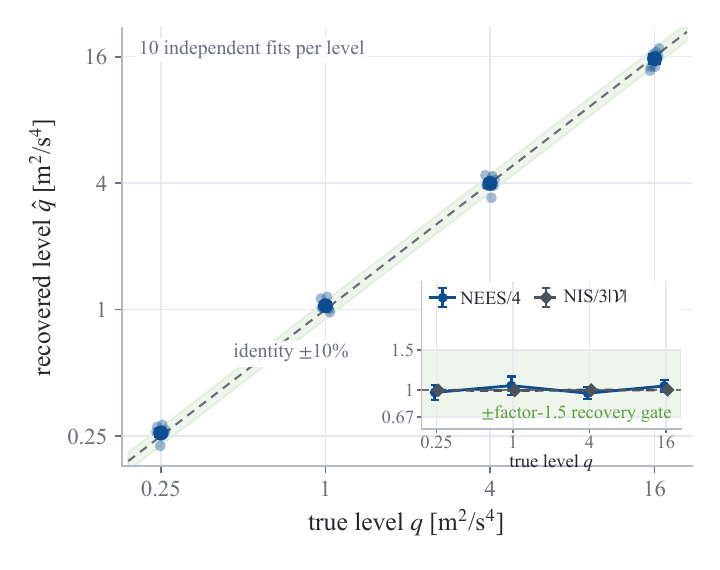}
\caption{Recovered versus tested process-covariance level. Markers show individual fits and per-level means with 95\% intervals; the dashed line denotes exact recovery. The inset reports normalized NEES and NIS against the prespecified recovery region.}
\label{fig:covariance_learning}
\end{figure}

\subsection{Process-Covariance Recovery and Calibration}
\label{subsec:native_radar_recovery}

Fig.~\ref{fig:covariance_learning} shows that the innovation-based criterion accurately recovers the tested process-covariance levels under model match. Across four held-out levels $q^\circ\in\{0.25,1,4,16\}$, the recovered-to-tested ratio was $0.977$--$1.045$, while normalized NIS and NEES were $0.997$--$1.004$ and $0.963$--$1.058$, respectively. The predeclared recovery criteria require the recovered covariance level to lie within \(\pm10\%\) of \(q^\circ\) and both normalized NIS and NEES to lie within the factor-\(1.5\) interval \([2/3,\,3/2]\). All four tested levels satisfied these criteria. The same calibration procedure was then applied to each motion-mismatch regime, and the resulting aggregate was fixed as \(\mathbf Q^\star\) throughout that regime's online evaluation.

The federated calibration is compared with a fixed hand-tuned baseline and a centralized pooled-data reference under identical trajectories, observations, and random seeds. The centralized reference uses the same innovation-based calibration objective but pools records from all operators. The baseline uses the isotropic level $q_{\rm HT}=4~\mathrm{m^2/s^4}$, corresponding to a slot-wise acceleration standard deviation of \(2~\mathrm{m/s^2}\), and this value is kept unchanged across all motion regimes. Statistical inference focuses on the heavy-mismatch regime, while the exact-NCV and gentle-mismatch regimes provide reference cases.

As shown in Fig.~\ref{fig:calibration_mismatch}, the calibration benefit becomes pronounced under heavy process-model mismatch. Calibration reduced the mean normalized NEES from \(3.855\) to \(1.110\) and increased 4-D state coverage from \(0.663\) to \(0.904\). The centralized reference shows a similar calibration trend across the tested regimes. Across the ten paired evaluation seeds, two-sided Wilcoxon signed-rank tests showed a smaller absolute deviation from unit normalized NEES after calibration (\(W=4,\;p=0.014\)), while the paired RMSE difference was not statistically resolved (\(W=21,\;p=0.557\)). These results indicate that calibration primarily improves the consistency of the reported tracking uncertainty rather than the point-estimation accuracy, which is the role of \(\mathbf Q^\star\) in the subsequent fusion and transmission design.

\begin{figure*}[!t]
\centering
\includegraphics[width=\linewidth]{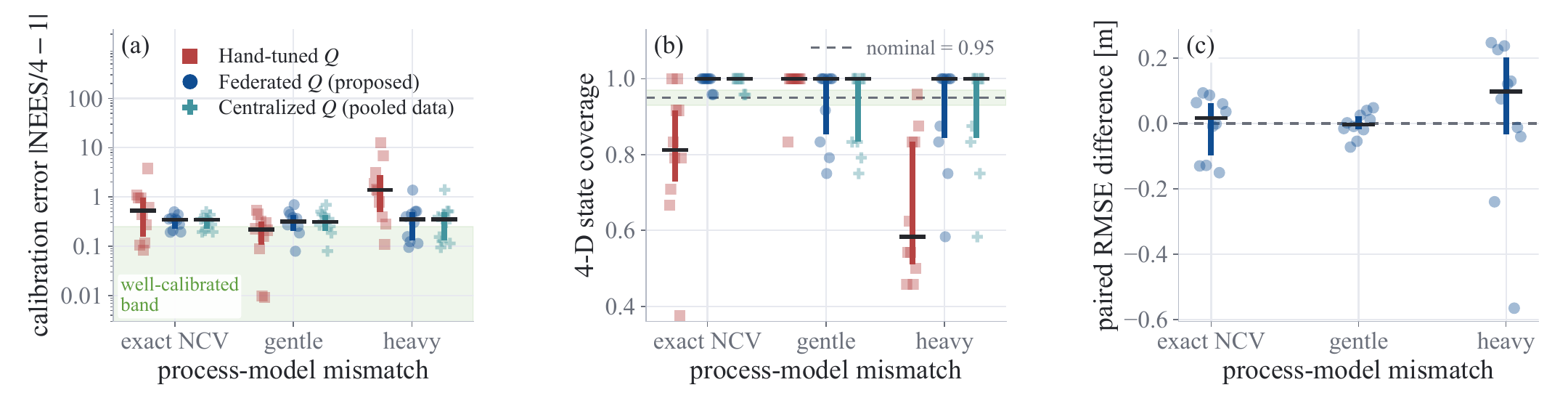}
\caption{Calibration under process-model mismatch. (a) Normalized-NEES error and (b) 4-D state coverage for the hand-tuned, federated, and centralized pooled-data calibrations; (c) paired RMSE difference between the federated and hand-tuned covariances. Faint markers denote evaluation seeds; bars show median/IQR. Dashed lines and shaded regions denote the corresponding reference targets.}
\label{fig:calibration_mismatch}
\end{figure*}

\begin{figure}[!t]
\centering
\includegraphics[width=\linewidth]{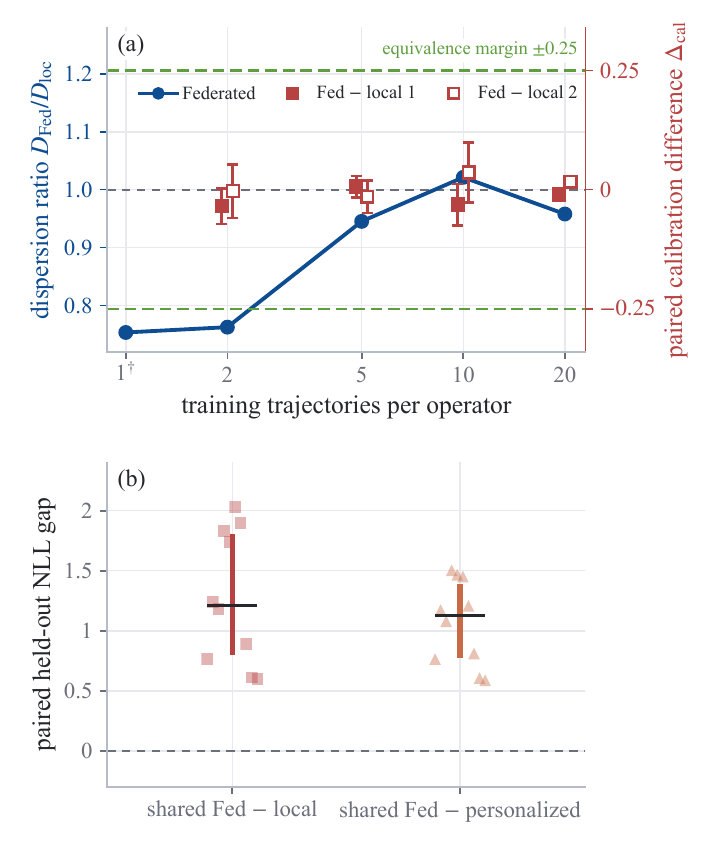}
\caption{Federation scope. (a) Federated-to-local dispersion ratio and paired calibration difference versus training trajectories per operator. Error bars show 90\% intervals; dashed bounds denote the calibration-equivalence margin. The dagger denotes a nonconverged run at one trajectory per operator. (b) Paired excess held-out NLL of the shared model relative to local and personalized calibration.}
\label{fig:federation_scope}
\end{figure}

\begin{figure}[!t]
\centering
\includegraphics[width=\linewidth]{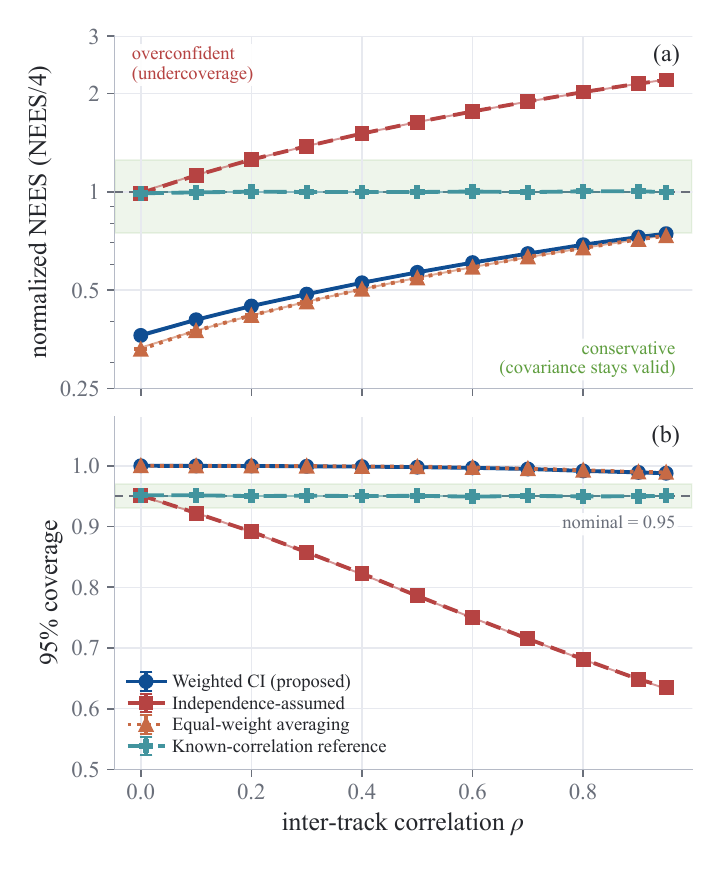}
\caption{Fusion under a fixed-marginal Gaussian correlation stress test. Panels show (a) normalized NEES and (b) 95\% state coverage as the shared-error correlation increases while each local marginal covariance is held fixed. Error bars denote 95\% intervals over ten seed-level means, with 500 samples per seed and condition.}
\label{fig:ci_correlated}
\end{figure}

\begin{figure*}[!t]
\centering
\includegraphics[width=\linewidth]{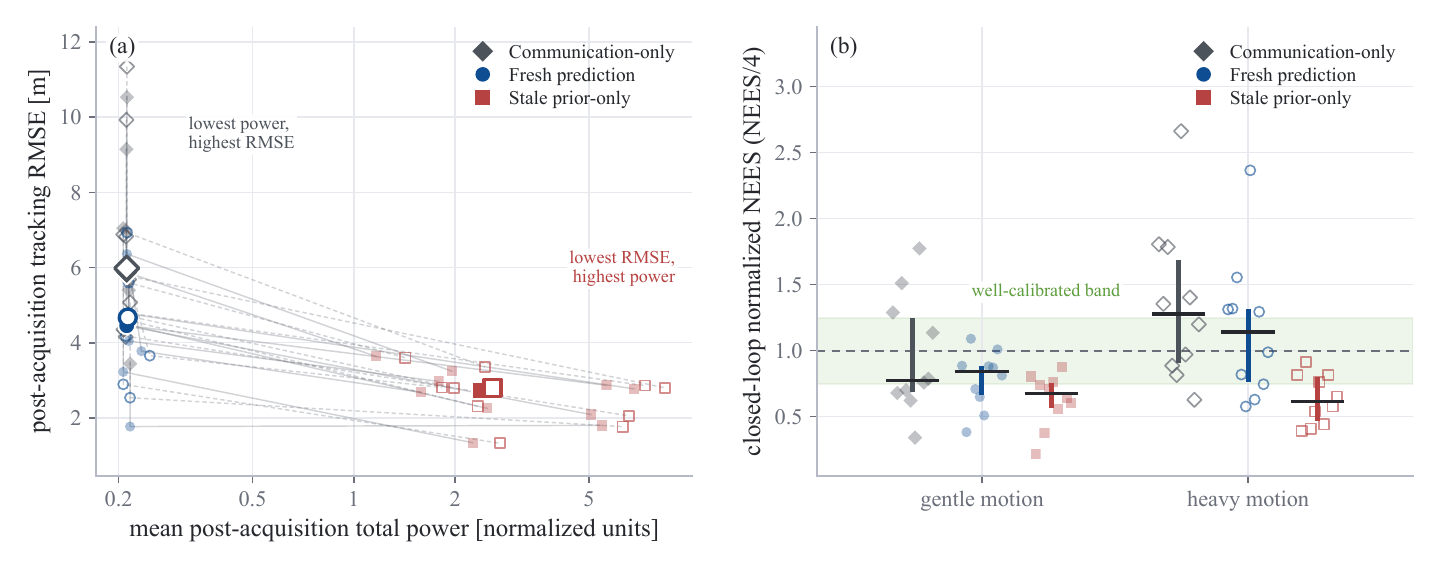}
\caption{Closed-loop performance: (a) post-acquisition power–RMSE operating points and (b) normalized NEES by controller input and motion regime. Faint markers denote seed-level summaries; large markers and bars show medians and IQRs.}
\label{fig:end_to_end_tradeoff}
\end{figure*}

\subsection{Federation Scope under Data Scarcity and Physical Heterogeneity}
\label{subsec:federation_closure}

Fig.~\ref{fig:federation_scope}(a) shows that federation primarily reduces the dispersion of the estimated process-covariance level when local data are scarce. Here, \(D\) denotes the sample standard deviation of the fitted process-covariance level across the ten evaluation seeds; \(D_{\rm Fed}\) is computed for the federated estimate, while \(D_{\rm loc}\) is the average of the corresponding dispersions from the two local estimates. The ratio $D_{\rm Fed}/D_{\rm loc}$ was $0.753$ and $0.762$ with one and two training trajectories per operator, respectively, and approached unity as local data increased. To assess whether this stabilization altered calibration quality, we define $\Delta_{\rm cal,j}=\lvert\operatorname{NEES}_{{\rm Fed},j}/4-1\rvert-\lvert\operatorname{NEES}_{{\rm loc},j}/4-1\rvert$. For two or more training trajectories per operator, both local comparisons remained within the prespecified equivalence margin $[-0.25,0.25]$, and paired two one-sided tests (TOSTs) confirmed calibration equivalence \cite{lakens2017equivalence}. Thus, the dispersion benefit under data scarcity is obtained without a resolved calibration penalty. The single run that did not meet the convergence criterion is retained in the reported dispersion statistic.

This benefit did not extend uniformly to strongly heterogeneous operators. Under strong operator heterogeneity, the median max-to-min empirical acceleration-variance ratio was $26.6$ across evaluation seeds. The complementary heterogeneity experiment in Fig.~\ref{fig:federation_scope}(b) evaluates the shared, local, and personalized calibrations using held-out innovation negative log-likelihood (NLL), normalized by the number of valid BS-time tuples. The shared-minus-local and shared-minus-personalized NLL gaps had medians of $1.212$ and $1.126$ with IQRs of \([0.798,1.806]\) and \([0.776,1.391]\), respectively. Both gaps were positive for all ten evaluation seeds, indicating consistently better held-out likelihood for the local and personalized calibrations. These results identify the practical boundary of a single shared covariance: federation is most useful when local data are scarce and motion statistics are sufficiently similar, whereas local or personalized calibration becomes more appropriate under strong heterogeneity.

\subsection{Track Fusion under Unknown Correlation}
\label{subsec:correlation_robust_fusion}

We evaluate fusion under increasing correlation between local tracking errors using a synthetic Gaussian stress test. The marginal error covariance of each local track is held fixed across all correlation levels, while the shared error component is varied to control the correlation between tracks. We compare weighted CI with independence-assumed fusion, equal-weight averaging, and a known-correlation reference. These respectively ignore the error cross-correlations, average the local means and marginal covariances, and use the true joint error covariance.

Fig.~\ref{fig:ci_correlated} shows the resulting normalized NEES and 95\% state coverage over ten independent evaluation seeds, with 500 samples per seed and correlation level. As the inter-track correlation increases, independence-assumed fusion becomes increasingly overconfident, with rising normalized NEES and decreasing coverage. In contrast, weighted CI remains conservative, with normalized NEES below unity and coverage above the nominal level throughout the tested range. The known-correlation reference remains close to the nominal consistency targets, while equal-weight averaging exhibits a similar conservative trend to CI. These results show that CI avoids the severe overconfidence caused by treating correlated tracks as independent, while preserving conservative covariance estimates without knowledge of the error cross-correlations.
\subsection{Closed-Loop Power--Tracking Performance}
\label{subsec:end_to_end_tradeoff}

We compare three controller inputs: Fresh uses the current fused prediction, Stale propagates the acquisition summary, and Communication-only removes the tracking-information constraint. Fig.~\ref{fig:end_to_end_tradeoff}(a) shows the resulting power--tracking operating points. Across both motion regimes, Fresh reduced RMSE by $31.0$--$31.7\%$ relative to Communication-only while increasing total power by only $0.9$--$1.7\%$. Stale achieved lower RMSE, but required $16.0\times$--$19.2\times$ the Fresh power. Table~\ref{tab:end_to_end_results} summarizes the corresponding closed-loop performance and numerical post-check results.

Table~\ref{tab:controller_input_ablation} isolates the source of the Fresh--Stale contrast through a $2\times2$ predictive mean or covariance ablation. Predictive-mean staleness produced no resolved change in power or RMSE, whereas covariance staleness increased total power by $3.22$--$4.02$ normalized units and reduced RMSE by $1.76$--$2.00$ m. The negligible mean effect cannot be attributed to nearly identical means: fresh and stale predictive means differed by $5.12$--$5.52$ m, while stale position-covariance traces were $6.55\times$--$7.21\times$ larger. The Fresh--Stale operating-point shift is therefore primarily driven by the freshness of the predictive covariance.

Fig.~\ref{fig:end_to_end_tradeoff}(b) shows that only two to six of ten seeds per controller–regime combination fell within the declared NEES band, so the aggregate medians should not be interpreted as calibration for every individual seed. At the median level, Fresh remained within or near the band, whereas Stale was systematically conservative.

SINR post-checks passed in all $220/220$ outcomes for every arm--regime cell. Tracking certification passed in $220/220$ and $219/220$ Fresh outcomes in the gentle and heavy regimes, respectively, compared with $197/220$ and $195/220$ for Stale. The largest tracking-constraint violation among uncertified slots was only \(1.97\times10^{-4}\,\mathrm{m}^2\), which is far below the threshold. Restricting the comparison to slots certified for both arms preserved the power contrast, with median Fresh-to-Stale power ratios of $0.087$ and $0.076$, respectively.
\begin{table}[!t]
\centering
\caption{POST-ACQUISITION CLOSED-LOOP PERFORMANCE.}
\label{tab:end_to_end_results}
\scriptsize
\setlength{\tabcolsep}{2pt}
\begin{tabular*}{\columnwidth}{@{\extracolsep{\fill}}cccccccc@{}}
\hline
Regime & Controller & \shortstack{Power\\(norm.)} & \shortstack{RMSE(m)} & NEES/4 & Coverage & \shortstack{SINR\\post-check} & \shortstack{Tracking\\post-check}\\
\hline
Gentle & Comm.-only & 0.2117 & 6.328 & 0.959 & 0.945 & 220/220 & N/A\\
Gentle & Fresh & 0.2136 & 4.323 & 0.780 & 1.000 & 220/220 & 220/220\\
Gentle & Stale & 3.4207 & 2.564 & 0.628 & 1.000 & 220/220 & 197/220\\
Heavy & Comm.-only & 0.2117 & 6.620 & 1.350 & 0.859 & 220/220 & N/A\\
Heavy & Fresh & 0.2154 & 4.569 & 1.160 & 0.895 & 220/220 & 219/220\\
Heavy & Stale & 4.1344 & 2.566 & 0.631 & 1.000 & 220/220 & 195/220\\
\hline
\end{tabular*}
\parbox{0.98\columnwidth}{\scriptsize \emph{Notes:} Continuous metrics are means of seed-level summaries over all planned slots. Checks report passed/attempted numerical certifications; N/A denotes no tracking constraint.}
\end{table}

\begin{table}[!t]
\centering
\caption{Paired controller-input ablation.}
\label{tab:controller_input_ablation}
\scriptsize
\setlength{\tabcolsep}{2pt}
\begin{tabular*}{\columnwidth}{@{\extracolsep{\fill}}llcc@{}}
\hline
Regime & Metric & Mean stale [95\% CI] & Cov. stale [95\% CI] \\
\hline
Gentle & Power (norm.) & $-0.008~[-0.078,\,0.062]$ & $+3.215~[1.697,\,4.733]$ \\
Heavy  & Power (norm.) & $-0.113~[-0.306,\,0.080]$ & $+4.021~[1.974,\,6.068]$ \\
Gentle & RMSE (m)      & $-0.002~[-0.015,\,0.011]$ & $-1.757~[-2.378,\,-1.135]$ \\
Heavy  & RMSE (m)      & $+0.001~[-0.016,\,0.017]$ & $-2.001~[-2.607,\,-1.395]$ \\
\hline
\end{tabular*}
\parbox{0.98\columnwidth}{\scriptsize
\emph{Notes:} Effects are relative to Fresh. “Stale mean only” uses \((\boldsymbol\mu_S,\mathbf P_F)\), whereas “Stale covariance only” uses \((\boldsymbol\mu_F,\mathbf P_S)\).}
\end{table}

\section{Conclusion}
\label{sec:conclusion}

This paper developed a framework for calibrated tracking uncertainty in multi-operator UAV ISAC. Local radar innovations calibrate the process covariance associated with unmodeled UAV motion, CI combines synchronized track summaries under unknown error correlations, and the fused prediction guides each operator's transmission design. Simulations showed improved uncertainty consistency under process-model mismatch and about a 31\% RMSE reduction relative to Communication-only with less than 2\% additional power; the crossed ablation further shows that the Fresh–Stale difference is driven mainly by the freshness of the predictive covariance. The formulation also admits exact rank-one recovery and conditional finite-interval bounds on local tracking covariances.

\appendix
\section{Supporting Analysis and Formal Proofs}
\label{app:formal_proofs}

\subsection{Proof of Lemma~\ref{lem:ci_prior_comparability}}
\label{app:proof_ci_prior_comparability}
\begin{proof}
Since \(\mathbf P_{j,t|t}\succ0\) for every \(j\in\mathcal A_t\), the information matrices in~\eqref{eq:information_form_track} are positive definite. By the definition of $\kappa_{j|c,t}$,
$\boldsymbol\Lambda_{j,t}
\preceq\kappa_{j|c,t}\boldsymbol\Lambda_{c,t}$.

Hence,
\[
\begin{aligned}
    \boldsymbol\Lambda_{f,t}(\boldsymbol\omega)
    =\sum_{j\in\mathcal A_t}\omega_j\boldsymbol\Lambda_{j,t}\preceq
    \left(\sum_{j\in\mathcal A_t}\omega_j\kappa_{j|c,t}\right)\boldsymbol\Lambda_{c,t}
    \preceq\rho_{c,t}\boldsymbol\Lambda_{c,t}.
\end{aligned}
\]
Since inversion reverses the Loewner order,
$\mathbf P_{c,t|t}\preceq\rho_{c,t}\mathbf P_{f,t|t}(\boldsymbol\omega)$.
Under the common prediction map and $\rho_{c,t}\geq1$,
\[
\begin{aligned}
    \mathbf P_{c,t+1|t}
    &=\mathbf F\mathbf P_{c,t|t}\mathbf F^{\mathsf T}+\mathbf Q^\star\\
    &\preceq\rho_{c,t}\mathbf F\mathbf P_{f,t|t}(\boldsymbol\omega)\mathbf F^{\mathsf T}+\mathbf Q^\star\\
    &\preceq\rho_{c,t}\mathbf P_{f,t+1|t}.
\end{aligned}
\]
Congruence with \(\mathbf B_{\rm pos}\) gives the corresponding position-block inequality.
\end{proof}

\subsection{Proof of \eqref{eq:constraint_receiver_bridge}}
\label{app:proof_constraint_receiver_bridge}
\begin{proof}
By the definition of $\mathbf J_{c_0,t}^{\rm rx}$,
\[
    \mathbf P_{c_0,t|t}^{\rm pos}
    =
    \left[
    (\mathbf P_{c_0,t|t-1}^{\rm pos})^{-1}
    +\mathbf J_{c_0,t}^{\rm rx}
    \right]^{-1}.
\]
Conditions (A2) and (A3) imply
\[
    (\mathbf P_{c_0,t|t-1}^{\rm pos})^{-1}
    +\mathbf J_{c_0,t}^{\rm rx}
    \succeq
    \rho_t^{-1}
    (\mathbf P_{f,t|t-1}^{\rm pos})^{-1}
    +\beta_t\mathbf J_{c_0,t}^{\rm des}.
\]
Hence,
\[
    (\mathbf P_{c_0,t|t-1}^{\rm pos})^{-1}
    +\mathbf J_{c_0,t}^{\rm rx}
    \succeq
    \frac{
    (\mathbf P_{f,t|t-1}^{\rm pos})^{-1}
    +\mathbf J_{c_0,t}^{\rm des}}
    {\max\{\rho_t,\beta_t^{-1}\}}.
\]
After inversion,
\[
    \mathbf P_{c_0,t|t}^{\rm pos}
    \preceq
    \max\{\rho_t,\beta_t^{-1}\}
    \left[
    (\mathbf P_{f,t|t-1}^{\rm pos})^{-1}
    +\mathbf J_{c_0,t}^{\rm des}
    \right]^{-1}.
\]
Using (A1) and the definition of $\bar\epsilon$ gives
$\operatorname{tr}
(\mathbf P_{c_0,t|t}^{\rm pos})
\leq \bar\epsilon.$

\end{proof}

\subsection{Proof of Proposition~\ref{prop:rank_one_recovery}}
\label{app:proof_rank_one_recovery}
\begin{proof}
Feasibility, $\gamma_{c,k}>0$, and $\sigma_{c,k}^2>0$ imply
\[
    d_{c,k,t}
    =\mathbf h_{c,k,t}^{\mathsf H}\mathbf W_{c,k,t}\mathbf h_{c,k,t}
    \geq\gamma_{c,k} I_{c,k,t}
    \geq\gamma_{c,k}\sigma_{c,k}^2>0.
\]
Let $\mathbf u_{c,k,t}=\mathbf W_{c,k,t}^{1/2}\mathbf h_{c,k,t}$. Then
$d_{c,k,t}=\|\mathbf u_{c,k,t}\|_2^2$ and
\[
\boldsymbol\Delta_{c,k,t}
=\mathbf W_{c,k,t}^{1/2}
\left(\mathbf I-
\frac{\mathbf u_{c,k,t}\mathbf u_{c,k,t}^{\mathsf H}}
{\|\mathbf u_{c,k,t}\|_2^2}\right)
\mathbf W_{c,k,t}^{1/2}\succeq\mathbf0.
\]
Moreover,
\[
    \mathbf h_{c,k,t}^{\mathsf H}\boldsymbol\Delta_{c,k,t}\mathbf h_{c,k,t}=0,
    \qquad
    \mathbf h_{c,k,t}^{\mathsf H}\bar{\mathbf W}_{c,k,t}\mathbf h_{c,k,t}=d_{c,k,t}.
\]
The transformed variables satisfy
$\bar{\mathbf W}_{c,k,t}=\mathbf W_{c,k,t}-\boldsymbol\Delta_{c,k,t}$ and
$\bar{\mathbf S}_{c,t}=\mathbf S_{c,t}+\sum_j\boldsymbol\Delta_{c,j,t}$, so the total
covariance is unchanged. For every user $k$, the transformed interference term
is
\begin{align*}
\bar I_{c,k,t}
&=\sum_{\ell\neq k}\mathbf h_{c,k,t}^{\mathsf H}
 (\mathbf W_{c,\ell,t}-\boldsymbol\Delta_{c,\ell,t})\mathbf h_{c,k,t}\\
&\quad+\mathbf h_{c,k,t}^{\mathsf H}
 \left(\mathbf S_{c,t}+\sum_j\boldsymbol\Delta_{c,j,t}\right)\mathbf h_{c,k,t}+\sigma_{c,k}^2\\
&=I_{c,k,t}+\mathbf h_{c,k,t}^{\mathsf H}\boldsymbol\Delta_{c,k,t}\mathbf h_{c,k,t}
=I_{c,k,t}.
\end{align*}
The desired signal power is also unchanged. Hence all SINR constraints are preserved. Since the total covariance is unchanged, the objective, per-BS power constraints, and all total-covariance tracking constraints are preserved.
\end{proof}

\subsection{Proof of Theorem~\ref{thm:covariance_containment}}
\label{app:proof_covariance_containment}
\begin{proof}
Fix \(s\in\{t_0,\ldots,T-1\}\), suppress \(c_0\), and let $q_\Sigma=q_x^\star+q_y^\star$,
$\mathbf P=\mathbf P_{s|s},\quad \mathbf P^-=\mathbf P_{s+1|s},\quad \mathbf P^+=\mathbf P_{s+1|s+1}\preceq\mathbf P^- .$
Define $\mathbf C \triangleq \mathbf B_{\rm vel} - {\mathbf B_{\rm pos}}/{\Delta t}$. From the NCV matrices in~\eqref{eq:ncv_block_matrices},
$$\mathbf C\mathbf F = -\frac{\mathbf B_{\rm pos}}{\Delta t}, \qquad \mathbf C\mathbf G = \frac{\Delta t}{2}\mathbf I_2,$$
and hence
$$\operatorname{tr} (\mathbf C\mathbf P^-\mathbf C^T) \le \frac{\bar\epsilon}{\Delta t^2} + \frac{\Delta t^2}{4}q_\Sigma.$$

Using the semi-norm $\|\mathbf X\|_{\mathbf P} \triangleq [\operatorname{tr}(\mathbf X\mathbf P\mathbf X^T)]^{1/2}$, covariance reduction, the triangle inequality, and the position bound at \(s+1\) give

$$ \begin{aligned} \|\mathbf B_{\rm vel}\|_{\mathbf P^+} &\le \|\mathbf C\|_{\mathbf P^+} + \frac{1}{\Delta t} \|\mathbf B_{\rm pos}\|_{\mathbf P^+}\\ &\le \|\mathbf C\|_{\mathbf P^-} + \frac{\sqrt{\bar\epsilon}}{\Delta t}\\ &\le \sqrt{ \frac{\bar\epsilon}{\Delta t^2} + \frac{\Delta t^2}{4}q_{\Sigma} } + \frac{\sqrt{\bar\epsilon}}{\Delta t}\\ &\le \frac{2\sqrt{\bar\epsilon}}{\Delta t} + \frac{\Delta t}{2}\sqrt{q_{\Sigma}}. \end{aligned} $$
This proves the velocity bound, and the position bound immediately gives
$$ \operatorname{tr}(\mathbf P^+) \le \bar\epsilon+\bar V. $$

For \(s\ge t_0+1\), by using Cauchy–Schwarz,
$$|\operatorname{tr}(\mathbf B_{\rm pos}\mathbf P\mathbf B_{\rm vel}^\mathsf{T})| \leq
\left[
\operatorname{tr}(\mathbf B_{\rm pos}\mathbf P\mathbf B_{\rm pos}^{\mathsf T})
\operatorname{tr}(\mathbf B_{\rm vel}\mathbf P\mathbf B_{\rm vel}^{\mathsf T})
\right]^{\frac{1}{2}}\leq \sqrt{\bar\epsilon\bar V}.$$
Expanding
\(\mathbf P^-=\mathbf F\mathbf P\mathbf F^T+\mathbf Q^\star\)
and applying the position, velocity, and cross-covariance bounds yields
$$ \operatorname{tr}(\mathbf P^-) \le \left( \sqrt{\bar\epsilon} + \Delta t\sqrt{\bar V} \right)^2 + \bar V + \left( \Delta t^2+\frac{\Delta t^4}{4} \right)q_\Sigma, $$
which proves~\eqref{eq:covariance_containment_bounds}.
\end{proof}

\bibliographystyle{IEEEtran}
\balance
\bibliography{references}

\vfill

\end{document}